\documentclass[amxtex]{amsart}
\usepackage{graphicx}
\usepackage{color}
\usepackage{amsmath}
\usepackage{amsfonts}
\usepackage{amssymb}
\usepackage{amscd}
\usepackage{amsthm}
\usepackage{hyperref}

\usepackage[normalem]{ulem}

\usepackage{bm}

\newtheorem{thm}{Theorem}
\newtheorem*{thm*}{Theorem}
\newtheorem{lem}[thm]{Lemma}

\newtheorem{assum}{Assumption}[section]
\newtheorem{rem}[thm]{Remark}
\newtheorem{defn}{Definition}
\newcommand{\lV}{\left\Vert}
\newcommand{\rV}{\right\Vert}

\newcommand{\lmk}{\left(}
\newcommand{\rmk}{\right)}

\newcommand{\Aut}{\mathop{\mathrm{Aut}}\nolimits}
\newcommand{\Tr}{\mathop{\mathrm{Tr}}\nolimits}

\newcommand{\Ad}{\mathop{\mathrm{Ad}}\nolimits}

\DeclareMathOperator*{\slim}{s-lim}

\newcommand{\norm}[1]{\lV #1 \rV}
\newcommand{\abs}[1]{\left\vert#1\right\vert}

\newcommand{\todo}[1]{{\color{red} To Do : #1}}

\newcommand{\caA}{{\mathcal A}}
\newcommand{\caB}{{\mathcal B}}
\newcommand{\caC}{{\mathcal C}}

\newcommand{\caF}{{\mathcal F}}

\newcommand{\caH}{{\mathcal H}}

\newcommand{\caM}{{\mathcal M}}

\newcommand{\caS}{{\mathcal S}}

\newcommand{\caU}{{\mathcal U}}
\newcommand{\caV}{{\mathcal V}}

\newcommand{\bbC}{{\mathbb C}}

\newcommand{\bbE}{{\mathbb E}}

\newcommand{\bbN}{{\mathbb N}}

\newcommand{\bbR}{{\mathbb R}}

\newcommand{\bbZ}{{\mathbb Z}}

\newcommand{\loc}{\mathrm{loc}}

\newcommand{\id}{\mathop{\mathrm{id}}\nolimits}
\newcommand{\unit}{\mathbb I}

\newcommand{\ep}[1]{\mathrm{e}^{#1}}

\newcommand{\ld}{\Lambda}

\title[Tensor categories and the quantum Hall effect]{Tensor category describing anyons in the quantum Hall effect and quantization of conductance}

\author{Sven Bachmann}
\address{Department of Mathematics \\ The University of British Columbia \\ Vancouver, BC V6T 1Z2 \\ Canada}
\email{sbach@math.ubc.ca}

\author{Matthew Corbelli}
\address{Department of Mathematics \\ UC Davis \\ Davis, 95616, CA \\ U.S.}
\email{mdcorbelli@ucdavis.edu}

\author{Martin Fraas}
\address{Department of Mathematics \\ UC Davis \\ Davis, 95616, CA \\ U.S.}
\email{mfraas@ucdavis.edu}

\author{Yoshiko Ogata}
\address{Research Institute for Mathematical Sciences\\ Kyoto University\\ Kyoto 606-8502\\ JAPAN}
\email{yoshiko@kurims.kyoto-u.ac.jp}

\date{\today}

\begin{document}
\maketitle
\begin{center}
Dedicated to the memory of Professor Huzihiro Araki
\end{center}
\begin{abstract}
In this study, we examine the quantization of Hall conductance in an infinite plane geometry. We consider a microscopic charge-conserving system with a pure, gapped infinite-volume ground state. While Hall conductance is well-defined in this scenario, existing proofs of its quantization have relied on assumptions of either weak interactions, or properties of finite volume ground state spaces, or invertibility. Here, we assume that the conditions necessary to construct the braided $C^*$-tensor category which describes anyonic excitations are satisfied, and we demonstrate that the Hall conductance is rational if the tensor category is finite.
\end{abstract}


\section{Introduction}

For an effectively two-dimensional system, such as a metal plate or a single graphene layer, the applied electric field and the induced current are two-component vectors. According to Ohm's law, for small fields, the current is proportional to the applied field. The matrix that relates them is called the conductance matrix. In an insulator, the current can only flow in the direction transversal to the applied field. The corresponding conductance matrix is antisymmetric, and Ohm's law takes the form 
\begin{equation}\label{Off diagonal Conductance}
\vec{J} = \begin{pmatrix}
0 & \kappa \\
-\kappa & 0
\end{pmatrix}
\vec{V}
\end{equation}
where we call the off-diagonal conductance $\kappa$ the Hall conductance.

The quantum Hall effect refers to the behaviour of $\kappa$ at low temperatures. As observed by Kitzling \cite{Kitzling} and Tsui, St\"{o}rmer and Gossard~\cite{Tsui}, whenever the material is insulating, i.e., the  conductance matrix is as in~(\ref{Off diagonal Conductance}), the Hall conductance is a fractional multiple of a universal constant.\footnote{We will use units in which this constant is equal to $(2 \pi)^{-1}$, and consequently, $2 \pi \kappa$ is a rational number.} The effect is called integer quantum Hall effect if the Hall conductance is a whole number and fractional quantum Hall effect if the Hall conductance is a non-integer rational number.

The integer quantum Hall effect is well modelled by non-interacting electrons in disordered media. The fact that $\kappa$ is integer-valued in this case is now reasonably well understood, and it is beyond the scope of this article to review the extensive body of literature on this topic. Let us mention that integer quantization remains true in the case of weak interactions~\cite{GMP} and under the additional assumption that the ground state is invertible~\cite{kapustin2020hall}. As a consequence, electron-electron interactions must be included to obtain a non-integer Hall conductance, which introduces significant analytical challenges. Consequently, the fractional quantum Hall effect is mathematically much less understood. A microscopic framework for a finite number of interacting electrons was already developed by Avron and Seiler in~\cite{AvronSeiler}, resulting in a possibly rational Hall conductance \cite{KleinSeiler}. A topological field theory of quantum Hall fluids in the bulk, which yields fractional quantization and anyonic excitations, was developed in th early 90's by Fr\"{o}hlich and collaborators, \cite{FrohlichKerler,FrohlichReview} and again later~\cite{FrohlichSchweigert}. An interacting microscopic framework with a well-defined thermodynamic limit was only provided twenty years later in the work of Hastings and Michalakis \cite{HastingsMichalakis}.

The setting of Hastings and Michalakis and of subsequent works \cite{RationalIndex, MBIndex,MonacoTeufel} involves a gapped Hamiltonian for interacting particles with a $U(1)$ symmetry on a finite torus of linear size $L$. Assuming that the Hamiltonian has $p$ locally indistinguishable ground states (along with some further technical assumptions), it is proved that
$$
2 \pi\kappa = \frac{q}{p} + O(L^{-\infty}),
$$ 
i.e. there exists $q \in \mathbb{Z}$ such that $|2 \pi \kappa - q/p|$ vanishes faster than any inverse power of $L$ as $L\to\infty$. This implies, see \cite{RationalIndex}, quantization of conductance in the plane, provided we assume that the ground state in the plane is a limit of ground states of embedded tori. Since they are locally indistinguishable it does not matter in the limit which torus ground states are used. This plausible assumption, often referred to as LTQO for Local Topological Quantum Order and introduced in \cite{TQO1,TQO2}, is likely satisfied in all standard quantum Hall models (in fact, it was forseen already in~\cite{WenNiu}) but it is currently difficult to prove, see however~\cite{YoungTQO} for recent progress in this direction.

In this article, we will show that Hall conductance is quantized in the infinite plane geometry without assuming LTQO. We want to avoid this assumption not due to the lack of proof -- we will anyway have to assume analytical properties we can't prove in any concrete model -- but because not having it leads to an intriguing intellectual puzzle: What replaces the ground states degeneracy on the finite torus in the denominator $p$ of the quantum Hall conductance fraction? We will show here that $p$ is upper bounded by the rank of the braided $C^*$-tensor category associated with the ground state \cite{MTC}, which describes the anyonic excitations in the system. A parallel approach was taken in~\cite{kapustin2020hall, SopenkoThesis}, where the infinite volume assumption is the invertibility of the state.

The connection between rational Hall conductance and the properties of low-energy excitations was first described in the works of Laughlin \cite{Laughlin,LaughlinWavefunction}, and Arovas, Schrieffer, Wilczek \cite{Arovas}. Laughlin demonstrated that  insertion of a $2 \pi$ flux produces an excitation with a fractional charge $2 \pi \kappa$ at the point of insertion. Arovas, Schrieffer, and Wilczek then showed that if a second excitation is adiabatically moved around the first, it acquires phase $e^{{i} (2 \pi)^2 \kappa}$. This means that the excitation is an Abelian anyon. In a finite volume setting that is very close to the present one, this was proved in~\cite{BBDFanyon}, and was extended to the infinite volume in~\cite{kapustin2020hall}. The connection exemplifies the interplay between macroscopic properties of a system, such as Hall conductance, and its microscopic properties, like the statistics of elementary excitations.

In this work, we use the theory developed by Doplicher, Haag and Roberts \cite{DHR,Haag} for relativistic quantum field theories, recently adapted to lattice systems \cite{N2015}, to describe anyon excitations. See the review \cite{LundholmReview} for other approaches to describing anyons. The DHR approach uses a superselection criterion to define excitation sectors, and proceeds to show that there is a natural braided $C^*$-tensor category structure associated with these sectors. In particular, physical elementary excitations correspond to objects in this category, and the physical braiding of two excitations corresponds to the braiding structure $\epsilon$ in the category. A complete mathematical setting in the context of quantum lattice systems was first described by Ogata \cite{MTC}, and we will use this particular framework here.

As mentioned above, the way how to construct Abelian anyons in fractional quantum Hall effect was introduced in \cite{BBDFanyon} and later expanded on and used to prove quantization for invertible systems in~\cite{kapustin2020hall}.  Neither of these works construct the anyons as objects of a braided $C^*$-tensor category. Firstly no exact framework existed at that time, and secondly (speaking for authors of \cite{BBDFanyon}) it seemed at the time that technical details associated with the precise construction might obscure the relatively simple idea behind the construction. We now feel that this has changed and that there is a need for uniform setting and precise definitions. The main technical part of this work, see Section~\ref{sec:anyon}, is the construction of some objects in the braided $C^*$-tensor category~$\mathcal{M}$ associated with the ground state. Echoing~\cite{Arovas} and~\cite{FrohlichKerler}, the braiding properties of these objects will be connected to the Hall conductance. In Section~\ref{sec:quantization}, we then prove that under the assumption that there is finite number of superselection sectors, Hall conductance $\kappa$ is indeed a rational number.

\section{Setting and results}
We follow the setting and notation of \cite{MTC}, which expands on the usual framework of 2-dimensional lattice spin systems. We consider a lattice $\mathbb{Z}^2$ and to each point $x \in \mathbb{Z}^2$ we associate an algebra $\mathcal{A}_{\{x\}}$ isomorphic to the algebra of $d \times d$ matrices for some fixed $d >1$. For a finite subset $\Gamma$ of $\mathbb{Z}^2$ we define $\mathcal{A}_\Gamma = \otimes_{x \in \Gamma} \mathcal{A}_{\{x\}}$. For $\Gamma_1 \subset \Gamma_2$, the algebra $\mathcal{A}_{\Gamma_1}$ is canonically embedded in $\mathcal{A}_{\Gamma_2}$ by tensoring operators in $\mathcal{A}_{\Gamma_1}$ with the identity. For infinite $\Gamma \subset \mathbb{Z}^2$, the algebra $\mathcal{A}_\Gamma$ is defined as an inductive limit of algebras associated with finite subsets of $\Gamma$. We denote $\mathcal{A} = \mathcal{A}_{\mathbb{Z}^2}$. For each $\Gamma\subset \bbZ^2$, we fix the conditional expectation 
$\bbE_\Gamma
: \caA\to \caA_\Gamma$ 
onto $\caA_\Gamma$
preserving the trace.
The algebra of local observables is denoted by $\caA_{\mathrm{loc}}$.

We will use notation, definitions and some results about interactions and dynamics that are summarized in Appendix~\ref{app:interactions}. While most of what we use should be standard for an expert in the field, the notion of an anchored interaction which was introduced in~\cite{BDFJ} might be an exception.

We consider an interaction $h \in \mathcal{J}$, here $\mathcal{J}$ is a class of interactions that are sufficiently local and uniformly bounded (see the appendix for the exact definition), and assume that it has a finite range, i.e. there exists $r >0$ such that $\mathrm{diam}(S) > r$ implies $h_S = 0$. We denote $\{\tau_t^h:t\in\bbR\}$ the dynamics, namely the one parameter group of automorphisms, generated by $h$.

\begin{assum}
\label{assum:1}
The dynamics $\tau^h$ has a unique gapped ground state $\omega$.
\end{assum}

Precisely, this means that there is a unique state $\omega$ satisfying
\begin{equation}\label{algebraic gap}
\frac{\omega(A^*[h,A])}{\omega(A^* A)}\geq g > 0
\end{equation}
for all local $A$ such that $\omega(A) = 0$. 
It is then automatically a ground state, i.e., $\omega(A^*[h,A])\geq 0$ for all local observables~$A$, and is pure
\cite{T}.
We denote the GNS representation of $\omega$ by $(\mathcal{H}, \pi, \Omega)$.

Note that we do not assume that $\omega$ is the unique state satisfying the condition $\omega(A^*[h,A])\geq 0$, namely there may in general be other such `algebraic ground states'.

\subsection{Braided $C^*$-tensor category associated with $\pi$}\label{sub:category}
In this section we recall, to the extent that we will need in this work, the construction of braided $C^*$-tensor category described in \cite{MTC}. It requires the approximate Haag duality. We do not present the full definition here and refer reader to \cite[Definition~1.1]{MTC}.

\begin{assum}
\label{assum:2}
The GNS representation $(\mathcal{H}, \pi,\Omega)$ of $\omega$ satisfies the approximate Haag duality.
\end{assum}

We denote $\bm e_\beta:=(\cos\beta,\sin\beta)$ and set
\begin{align}
\label{eq:cones}
\ld_{\bm a,\theta,\varphi}:=\{\bm a+t \bm e_{\beta}\mid t > 0,\: \beta\in (\theta-\varphi,\theta+\varphi)\}
\end{align}
for $\theta\in\bbR$, $\bm a \in \mathbb{R}^2$, and $\varphi\in (0,\pi)$.
We call a subset of this shape a cone and use the same notation for the subset $\ld_{\bm a,\theta,\varphi}\cap\bbZ^2$ of the lattice. It is important that the empty set and $\mathbb{R}^2$ are not cones. Strict Haag duality is the statement that $\pi(\caA_{\Lambda^c})' = \pi(\caA_\Lambda)''$ for all cones $\Lambda$ while the approximate version allows for `tails' on the outside of the cones.  

We now define superselection sectors with respect to the GNS representation $(\mathcal{H}, \pi,\Omega)$ of the gapped ground state
$\omega$, see Assumption~\ref{assum:1}. We note that the representation is irreducible because the ground state $\omega$ is pure.

\begin{defn}
We say that a representation $\sigma$ of $\mathcal{A}$ on $\mathcal{H}$ satisfies the \emph{superselection criterion} with respect to $\pi$ if 
$$
\sigma |_{\mathcal{A}_{\Lambda^c}} \simeq \pi |_{\mathcal{A}_{\Lambda^c}},
$$
for any cone $\Lambda$.
Here, $\simeq$ denotes unitary equivalence.
\end{defn}
We denote by $O$ all representations of $\caA$ on $\caH$ that satisfy the superselection criterion. Equivalence of representations splits $O$ into equivalence classes, which are called superselection sectors.

\begin{thm*}[\cite{MTC}, Theorems 5.2 \& 6.1] 
Given Assumptions~\ref{assum:1}, \ref{assum:2}, the superselection sectors form a braided $C^*$-tensor category.
\end{thm*}  
We call this category $\caM$ and refer to \cite{MTC} for precise definitions. We will recall the construction in Section~\ref{sec:MTC}. For the moment we only note that objects in the category are representations satisfying the superselection criteria, and morphisms are their intertwiners. The braiding, $\epsilon(\rho, \sigma)$ of objects $\rho, \sigma$, encodes the exchange statistics of the anyons corresponding to $\rho, \sigma$. We will also introduce a braiding statistics $\theta(\rho, \sigma)$ which will be the phase obtained by moving $\sigma$ counterclockwise around $\rho$, see~(\ref{eq:theta}).

\subsection{Charge conservation}
We consider an on-site $U(1)$ symmetry generated by an interaction $q \in \mathcal{J}$ such that operators $q_{\{x\}} \in \mathcal{A}_{ \{x\}}$ have integer spectrum for all $x \in \mathbb{Z}^2$, and $q_S = 0$ if $S$ is not a singleton. The operator $q_{\{x\}}$ encodes physical charge at site $x$, and for any finite region $\Gamma$ we denote
\begin{equation*}
Q_\Gamma := \sum_{x \in \Gamma} q_{\{x\}},
\end{equation*}
and refer to it as the \emph{charge} in the set $\Gamma$. By assumption,
$\mathrm{Spec}(Q_\Gamma)\subset\bbZ$.
For any (finite or not) subset $\Gamma$, let $\delta^q_\Gamma$ be the derivation associated with $q|_\Gamma$, the restriction of $q$ to $\Gamma$ --- see Appendix~\ref{app:restrictions} for the notion of restriction of an interaction --- and let $\alpha^\Gamma$ be the corresponding family of automorphisms. Note that $\alpha_{2\pi}^\Gamma = \mathrm{id}$, justifying the name $U(1)$ symmetry. We denote $\delta^q = \delta^q_{\mathbb{Z}^2}$, and $\alpha = \alpha^{\mathbb{Z}^2}$.

We assume that our system is $U(1)$ invariant in the following sense.
\begin{assum}
\label{assum:3}
For any finite  $S,\Gamma \subset \bbZ^2$ such that $S\subset \Gamma$,
\begin{equation}\label{charge conservation commutators}
[h_S,Q_\Gamma] = 0.
\end{equation}
\end{assum}

We immediately note that in conjunction with Assumption~\ref{assum:1}, this implies the $U(1)$-invariance of the state, namely $\omega \circ \alpha_\phi = \omega$ for all $\phi \in \mathbb{R}$.

Assumptions~\ref{assum:1}, \ref{assum:3} allow to construct a self-adjoint operator $J \in \mathcal{A}$ whose expectation value
\begin{equation}
\label{eq:Hall}
\kappa := \omega(J)
\end{equation}
is the Hall conductance of the system \cite{HastingsMichalakis, QHE1}. We provide details of this construction in Section~\ref{sec:Hall}. An alternative construction of an observable corresponding to Hall conductance is given in~\cite{SopenkoThesis} using the framework of higher Berry curvature~\cite{HigherBerry}.

\subsection{Results}
The first theorem that we will prove makes an explicit connection between the braided $C^*$-tensor category, specifically the braiding statistics $\theta(\rho,\rho)$ briefly introduced at the end of Section~\ref{sub:category} and defined in~(\ref{eq:theta}) below, and the Hall conductance $\kappa$. As discussed in the introduction, versions of this theorem are in \cite{BBDFanyon, kapustin2020hall}.

\begin{thm}[Existence of Anyons]
\label{thm:anyon}
Given Assumptions~\ref{assum:1} -- \ref{assum:3}, there exists a simple object $\rho \in \mathcal{M}$ such that 
\begin{equation*}
\theta(\rho, \rho) = e^{-{i} (2 \pi)^2 \kappa}.
\end{equation*}
\end{thm}

\noindent The second theorem that we prove addresses quantization of the Hall conductance.

\begin{thm}[Quantization of Hall conductance]
\label{thm:hall}
Suppose Assumptions~\ref{assum:1} -- \ref{assum:3} hold, and assume that there is a finite number $p'$, of equivalence class of simple objects in $\caM$. Then there exists an integer $p \leq p'$  such that 
$$
2\pi \kappa \in \mathbb{Z}/p.
$$
\end{thm}

\subsection{Outline} In the following section we provide details about construction of the braided $C^*$-tensor category. In Section~\ref{sec:Hall} we define the Hall conductance. In Section~\ref{sec:anyon} we prove Theorem~\ref{thm:anyon}, and in Section~\ref{sec:quantization} we prove Theorem~\ref{thm:hall}. Finally, Appendix~A contains all we need about interactions and their associated objects, and Appendix~B has some technical parts related to the definition of the braiding statistics $\theta$ on the braided tensor category.

\section{Construction of braided $C^*$-tensor category}
\label{sec:MTC}
The idea to use superselection sectors to describe anyon ground state excitations was first described in the context of algebraic quantum field theory in \cite{AnyonDHR}. It was recently adapted to quantum spin systems \cite{NaaijkensEndomorphisms, N2015, CNN, MTC}. A representation $\sigma$ that is quasi-equivalent to $\pi$, without any restriction, corresponds to local excitations of the ground state. A representation $\sigma$ that satisfies the superselection criteria but is not quasi-equivalent to $\pi$ corresponds to anyon excitation: We often visualize them as excitations created by an endomorphism acting along a string going from the point of the excitations to infinity, which is, in particular, localized inside a cone. This is the case in some exactly solvable models \cite{QDouble,KitaevHoneycomb}, see~\cite{NaaijkensEndomorphisms,N2015}.

In order to construct the braided $C^*$-tensor category, we shall now make various choices but the resulting category is independent of these choices \cite{MTC}. Let $\caC$ be the set of all cones~(\ref{eq:cones}) such that $[\theta-\varphi,\theta+\varphi]\cap [\frac {3\pi}2-\frac\pi 4, \frac {3\pi}2+\frac\pi 4]=\emptyset$ $\mod 2\pi$. This makes a choice for what is called the forbidden  direction, see Figure~\ref{Fig:Cones}. Let
$$
\mathcal{B}
:=\overline{\cup_{\ld\in \caC}\pi\lmk\caA_{\ld}\rmk'' },
$$
where the overline indicates the norm closure.
For each cone $\ld$ and $\sigma \in O$, we 
set
\begin{align*}
\caV_{\sigma,\ld}:=\{ V_{\sigma,\ld}\in \caU(\caH)\mid \Ad (V_{\sigma,\ld}) \circ \sigma\vert_{\caA_{\ld^c}}
=\pi\vert_{\caA_{\ld^c}}\},
\end{align*}
which is a nonempty set by the very definition of~$O$. We also denote by $O_\ld$ the set of all $\sigma\in O$ with $\unit\in \caV_{\sigma,\ld}$. $O_\ld$ represents anyonic excitations supported in $\Lambda$.

We fix a cone $\ld_0:=\ld_{\bm 0,\frac{\pi }2,\frac{5\pi} 8}\in \caC$, the objects in the category $\mathcal{M}$ are the elements of $O_{\ld_0}$. In order to introduce a tensor product of objects (and later braiding), we first pull $\sigma\in O_{\ld_0}$ to a map on the algebra $\caB$. There exists a unique *-homomorphism $T_\sigma$ of $\caB$ such that
\begin{equation*}
T_{\sigma}\circ\pi=\sigma
\end{equation*}
and $T_{\sigma}$ is weakly continuous on $\pi\lmk\caA_{\ld}\rmk'' $, for every $\ld\in \caC$.

For two objects $\sigma_1, \sigma_2 \in O_{\ld_0}$, their tensor product is defined as
$$\sigma_1\otimes \sigma_2:=T_{\sigma_1} \circ T_{\sigma_2} \circ \pi.$$
The morphisms of $\caM$ are the intertwiners
$$
\mathrm{Hom}(\sigma_1, \sigma_2) := \{V \in B(\mathcal{H}) \mid V \sigma_1 = \sigma_2 V \}.
$$
To define braiding we fix the following two cones $\ld_2:=\ld_{\bm 0, \pi, \frac\pi 8}, \ld_1:=\ld_{\bm 0,\frac\pi 2, \frac\pi 8}$.
For $\rho \in O_{\Lambda_1}$ and $\sigma \in O_{\Lambda_0}$ the braiding $\epsilon(\rho,\sigma)$, of $\rho, \sigma$ is defined as the norm limit
$$
\epsilon(\rho,\sigma) := \lim_{s\to\infty }V_{\sigma,\ld_2(s)}^* T_{\rho}
\lmk V_{\sigma,\ld_2(s)}\rmk.
$$
Here and later, we use a notation $\ld_{\bm a,\theta,\varphi}(s) = \ld_{\bm a,\theta,\varphi} + s{\bm e}_\theta$.
The braiding is independent of the choice of unitaries $V_{\sigma,\ld_2(s)} \in \caV_{\sigma,\ld_2(s)} $, and it intertwines $\rho \otimes \sigma$ with $\sigma \otimes \rho$, i.e.  $\epsilon(\rho, \sigma) \in \mathrm{Hom}(\rho \otimes \sigma, \sigma \otimes \rho)$. If $\rho = \pi$, then $T_\rho = \mathrm{id}$ and hence $\epsilon(\pi,\sigma) = 1$ for all $\sigma$. 

\begin{figure}
\includegraphics[width = 0.6\textwidth]{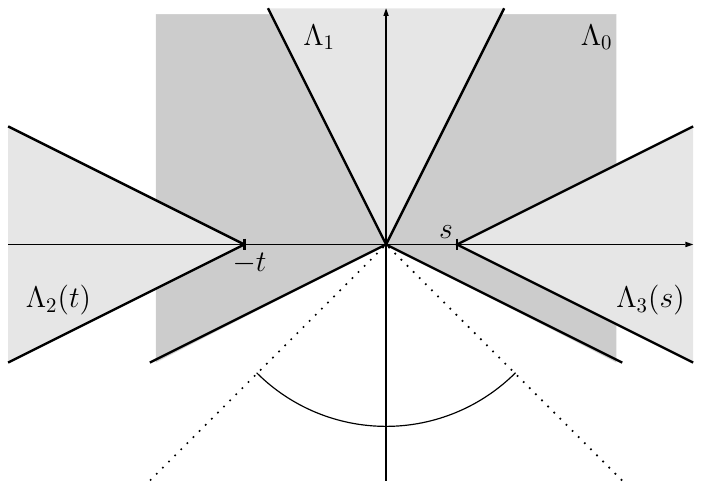}
\caption{The various cones used in the construction of the category $\caM$. Forbidden directions are represented by the arc in the lower half plane.}
\label{Fig:Cones}
\end{figure}

We will also need the braiding statistics, $\theta(\rho, \sigma)$, associated with winding of the anyon $\sigma$ around $\rho$. We fix another cone $\Lambda_3 =\Lambda_{\bm 0,0,\frac\pi 8}$, and for $\rho \in O_{\Lambda_1}$ and $\sigma \in O_{\Lambda_0}$, we define
\begin{equation}
\label{eq:theta}
\theta(\rho, \sigma)  = \lim_{t \to \infty} \epsilon(\rho, \Ad V_{\sigma, \Lambda_3(t)} \circ \sigma ).
\end{equation}
The limit is well defined, see (i) of the next lemma. In this article we will only encounter Abelian anyons in which case the braiding statistics is proportional to identity: This is reflected in the assumptions and statements of the following lemma. While a `braiding statistics' or `statistical phase' has been defined in many different ways in the mathematical literature and expresses the same phenomenology (among the close analogs, see Section~2 in~\cite{FRS}, Section~2 in~\cite{Rehren} or Section~8.5 in~\cite{HalvorsonMueger}), the authors are not aware of Definition~(\ref{eq:theta}) having appeared before.
\begin{lem}
\label{lem:theta}
Suppose that Assumptions \ref{assum:1}, \ref{assum:2} hold. Let $\rho, \sigma \in O_{\Lambda_1}$. Suppose that $\sigma$ is of the form $\sigma = \pi \circ \tilde{\sigma}$ for some $\tilde{\sigma} \in \Aut(\caA)$ such that $\tilde{\sigma}|_{\caA_{\Lambda_1^c}} = \id_{\caA_{\Lambda_1^c}}$. Then 

\begin{enumerate}
\item $\theta(\rho, \sigma)$ is well defined, and independent of the choice of $V_{\sigma, \Lambda_3(t)}$,

\item $\theta(\rho, \sigma) \in \mathrm{Hom}(\rho, \rho)$.
\end{enumerate}
Suppose in addition that $\rho = \pi \circ \tilde{\rho}$ for some automorphism $\tilde{\rho}$. Then

\begin{enumerate}
\item[(iii)]  $\theta(\rho, \sigma) = e^{i \theta} \id$, for some $\theta \in \mathbb{R}$,

\item[(iv)] For $\rho' = \mathrm{Ad}_V \circ \rho \in O_{\Lambda_1}$ and $\sigma' = \mathrm{Ad}_W \circ \sigma \in O_{\Lambda_0}$
$$
\theta(\rho', \sigma') = \theta(\rho, \sigma),
$$

\item[(v)] $\theta( \rho_1 \otimes \rho_2, \sigma) = \theta(\rho_1, \sigma) \theta(\rho_2, \sigma)$.

\end{enumerate}
\end{lem}
\begin{proof}
The proof of (i) is quite technical and similar to the proofs of existence of $\epsilon(\rho, \sigma)$ in \cite{MTC}. We postpone it to Appendix~\ref{app:theta}. Since similar techniques are required for the proof of part (ii), we similarly postpone it, see  Lemma~\ref{thetaRhoSigmaIsHomRhoRho}. 

Since $\rho$ is irreducible by the additional assumption, the point (ii) implies that $\theta(\rho, \sigma)$ is proportional to identity. Because $T_\rho$ is a unital $*$-endomorphism, $\theta(\rho,\sigma)$ is a unitary as the norm limit of a family of unitaries. It follows that $\theta(\rho, \sigma)$ is a phase. This proves (iii).

Manifestly, $\theta(\rho, \sigma') = \theta(\rho, \sigma)$. So to prove (iv), it remains to compute $\theta(\rho', \sigma)$. Let $\sigma'_s = \Ad(V_{\sigma,\Lambda_3(s)})\circ \sigma$. Pick $V_{\sigma'_s,\Lambda_2(t)}=V_{\sigma,\Lambda_2(t)}V_{\sigma,\Lambda_3(s)}^*$.
    \begin{align*}
        \theta(\rho',\sigma)&=\lim_{s\to\infty}\epsilon(\Ad(V)\circ\rho,\Ad(V_{\sigma,\Lambda_3(s)})\circ\sigma)\\
        &=\lim_{s\to\infty}\lim_{t\to\infty}[[V_{\sigma'_s,\Lambda_2(t)}^*,V]]\Ad(V)(\epsilon(\rho,\sigma'_s)) \\
        &=\lim_{s\to\infty}\lim_{t\to\infty}[[V_{\sigma,\Lambda_3(s)}V_{\sigma,\Lambda_2(t)}^*,V]]\Ad(V)(\epsilon(\rho,\sigma'_s)).
    \end{align*}
where we used Lemma~\ref{epsilonRhoPrimeSigma} in the second equality and denote $[[U_1,U_2]] = U_1U_2U_1^* U_2^*$ for the commutator of unitaries. For $A \in \caA_{\Lambda_1^c}$, we have that
\begin{equation*}
\pi(A) = \rho'(A) = \Ad(V)(\rho(A))=\Ad(V)(\pi(A)),
\end{equation*}
namely $V \in \pi(\caA_{\Lambda_1^c})'$.
Hence $\lim_{s\to\infty}\lim_{t\to\infty}[[V_{\sigma,\Lambda_3(s)}V_{\sigma,\Lambda_2(t)}^*,V]] = 1$ by Lemma~\ref{VSigmaLambdaThreeTwoCommutatorLambda1CompPrimeLimit} and we get
 $$
 \theta(\rho', \sigma) = \Ad(V) (\theta(\rho, \sigma)).
 $$
With this, Part (iv) follows from (iii).
 
To prove (v), we recall \cite{MTC} that for $\sigma_1,\sigma_2,\sigma_3\in O_{\ld_0}$,
$$
\epsilon\lmk\sigma_1\otimes \sigma_2,\sigma_3\rmk
=\epsilon(\sigma_1,\sigma_3)T_{\sigma_1}\lmk\epsilon(\sigma_2,\sigma_3)\rmk.
$$
From the definition of $\theta$, we then get  
$$
\theta(\rho_1 \otimes \rho_2, \sigma) = \theta(\rho_1, \sigma) T_{\rho_1} \lmk \theta(\rho_2, \sigma) \rmk,
$$
and claim then follows again from (iii).
\end{proof}

\section{Definition of Hall conductance}
\label{sec:Hall}
There are various, equivalent, formulas for Hall conductance. These formulas fall into two classes, the first class expresses Hall conductance as the adiabatic curvature of a certain ground state bundle. The second class expresses Hall conductance as a charge pumped upon insertion of a $2 \pi $ flux. The formulas can be proved from the Kubo formula \cite{ExactLR}, so the starting point is a matter of taste. We decided to start with a formula from the first class because it is most naturally formulated in the infinite volume limit. However, in the process of proving our main theorems we will need a formula from the second class which we will establish as a lemma below.

To define Hall conductance, we use a partition of space in four quadrants,
\begin{align*}
\begin{split}
&A:=\{(x,y)\in \bbZ^2\mid 0\le x, 0\le y\},\\
&B:=\{(x,y)\in\bbZ^2\mid x\le -1, 0\le y\},\\
&C:=\{(x,y)\in\bbZ^2\mid x \leq -1, y\le -1,\},\\
&D:= \{(x,y)\in\bbZ^2\mid 0 \le x, y\le -1\},\\
\end{split}
\end{align*}
see Figure~\ref{fig:quadrants}. We will use the notation $\Gamma_1 \Gamma_2 = \Gamma_1 \cup \Gamma_2$ for any two sets $\Gamma_1, \Gamma_2$. For example, $AB$ is the upper half plane. In addition, for a set $\Gamma$, we set $\Gamma_N:=\Gamma\cap[-N,N]^{\times 2}$.
\begin{figure}
\includegraphics[width=0.4\textwidth]{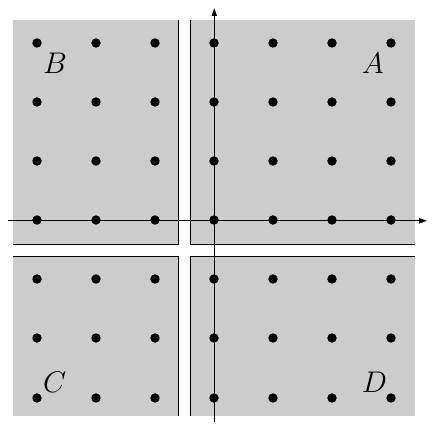}
\caption{The four quadrants used to define the Hall conductance}
\label{fig:quadrants}
\end{figure}

For a region $\Gamma \subset \mathbb{Z}^2$, we define
\begin{equation}\label{Def of k}
k^\Gamma = - \int dt W(t) \tau^h_t(\delta_\Gamma^q(h)),
\end{equation}
with $W(t)$ a  super-polynomially decaying function such that $i \sqrt{2 \pi} \hat{W}(k) = 1/k$ for $|k| \geq g$. Here, we use the specific definition of an interaction given in~(\ref{eq:interaction_evolution}), and so $k^\Gamma$ is a bonafide interaction. Next lemma gives basic properties of this interaction. We recall that the Hamiltonian has finite range~$r$, and let $\partial \Gamma := \{x \in \mathbb{Z}_2 : \mathrm{dist}(x, \Gamma) \leq r, \mathrm{dist}(x, \Gamma^c) \leq r\}$.

\begin{lem}
\label{lem:k}
Suppose that Assumptions~\ref{assum:1} and \ref{assum:3} hold. Then for any $\Gamma \subset \mathbb{Z}^2$,
\begin{enumerate}
\item $k^\Gamma$ is anchored in $\partial \Gamma$,
\item $\delta^q(k^\Gamma_S) = 0$,
\item $k^\Gamma = - k^{\Gamma^c}$.
\end{enumerate}
\end{lem}
\begin{proof}
The statement (i) follows from Lemma~\ref{lem:anchoring}. Note the non-trivial definition of time evolution of an interaction given in (\ref{eq:interaction_evolution}). Charge conservation, Assumption~\ref{assum:3}, implies (ii,iii).
\end{proof}

A consequence of (i) and Lemma~\ref{lem:inner_commutator} is that $i [k^{AB}, k^{AD}]$ is summable. With this, we can define the Hall conductance via (\ref{eq:Hall}) with
\begin{equation}\label{def:J}
J = \sum_S i [k^{AB}, k^{AD}]_S.
\end{equation}
It is in no way apparent that expectation value of $J$ is adiabatic curvature of some ground state bundle, we refer reader to \cite{QHE1} for the details about this bundle.

As announced in the first paragraph, we will need to connect this definition to a different formula that we will use later. We are going to do this in the remaining part of this section.

 For a region $\Gamma$ we define an interaction $\bar{q}^\Gamma = q|_\Gamma - k^\Gamma$, and denote $\beta^\Gamma$ the associated family of automorphisms. $\beta_{\phi}^\Gamma$ corresponds to threading flux $\phi$ through the boundary of $\Gamma$, see \cite{BBDFanyon}. The function $W$ in~(\ref{Def of k}) is chosen so that the state $\omega$ is invariant, namely
\begin{equation}\label{qbar invariance}
\omega\circ\delta^{\bar{q}^\Gamma} = 0,\qquad \omega\circ\beta^\Gamma_\phi = \omega
\end{equation}
for all regions $\Gamma$, see~\cite{Adiabatic}. For finite $\Gamma$, the operator $K_\Gamma = \sum_{S} k^{\Gamma}_S$ is well-defined and the invariance above can be phrased as
\begin{equation}\label{qbar finite}
\omega([\bar{Q}_\Gamma,A]) = 0
\end{equation}
for all $A\in\caA$, where $\bar{Q}_\Gamma := Q_\Gamma - K_\Gamma$.

We claim that the interaction $i [\bar{q}^{AB}, k^{AD}]$ is also summable and that
$$
\omega(\sum_S i [\bar{q}^{AB}, k^{AD}]_S) = 0.
$$
To establish this, we start by recalling that $k^{AD}$ is anchored in $\partial (AD)$, see Lemma~\ref{lem:k}(i). We now split the sum to two parts. First of all, if $S \subset AB$, we have that
$$
[\bar{q}^{AB}, k^{AD}]_S = \sum_{S_1 \cup S_2 = S} [\bar{q}^{AB}_{S_1}, k^{AD}_{S_2}] =- \sum_{S_1 \cup S_2 = S} [k^{AB}_{S_1}, k^{AD}_{S_2}],
$$
by Lemma~\ref{lem:k}(ii). By Lemma~\ref{lem:inner_commutator}, the sum $
\sum_{S_1 \cup S_2 \subset AB} [k^{AB}_{S_1}, k^{AD}_{S_2}]
$
is absolutely convergent, in particular, we can write it as 
$$
-\sum_{S_2 \subset AB} \sum_{S_1 \subset AB} [k^{AB}_{S_1}, k^{AD}_{S_2}] = \sum_{S_2 \subset AB} \sum_{S_1 \subset AB} [\bar{q}^{AB}_{S_1}, k^{AD}_{S_2}],
$$
where we used that for all $S_2 \subset AB$, $\sum_{S_1 \subset AB} [q_{S_1}, k^{AD}_{S_2}] = \delta^q(k^{AD}_{S_2})=0$. It might be worth noting that the double sum on the RHS is not absolutely convergent anymore. Second of all, we consider those sets $S$ that intersect $(AB)^c$. In fact, the anchoring of $k^{AD}$ implies that we are considering only those that intersect both $\partial(AD)$ and $(AB)^c$. On the one hand, the sum
\begin{equation*}
\sum_{S_1, S_2 : (S_1 \cup S_2) \cap \partial(AD) \cap (AB)^c \neq \emptyset} [k^{AB}_{S_1}, k^{AD}_{S_2}]
\end{equation*}
is absolutely convergent by Lemma~\ref{lem:inner_commutator}. On the other hand, since the interaction $q$ is strictly on site, 
$$
\sum_{S : S \cap \partial(AD) \cap (AB)^c \neq \emptyset} [q^{AB}, k^{AD}]_S = \sum_{S_1,S_2:S_2 \cap \partial(AD) \cap (AB)^c  \neq \emptyset} [q^{AB}_{S_1}, k^{AD}_{S_2}]
$$
and the sum on the RHS is absolutely convergent. Altogether, we have now established that the commutator is summable. The above argument also yields that if the two sums are put added to each other, we obtain a convergent sum,
\begin{align}\label{panda}
\sum_S i[\bar{q}^{AB}, k^{AD}]_S = \sum_{S_2} \delta^{\bar{q}^{AB}}(k^{AD}_{S_2}).
\end{align}
It follows that the expectation vanishes since $\omega(\delta^{\bar{q}^{AB}}(k^{AD}_{S_2})) = 0$ for every $S_2$. By the same argument, the equality also holds with $AB$ and $AD$ exchanged. So we established two new expressions for Hall conductance,
$$
\omega(\sum_{S} [k^{AB}, k^{AD}]_S) = \omega(\sum_S [q^{AB}, k^{AD}]_S) = \omega(\sum_S [k^{AB}, q^{AD}]_S).
$$
Adding the last two, and subtracting the first and a zero $\sum_S [q^{AB}, q^{AD}]_S$ we then get
\begin{equation}
\label{eq:barq_hall}
\omega(J) =  \omega(-i \sum_{S} [\bar{q}^{AB}, \bar{q}^{AD}]_S),
\end{equation}
with $J$ defined in~(\ref{def:J}). To avoid any confusion, we remark that the expectation on the RHS looks formally zero by (\ref{qbar invariance}).  However, the double sum $\sum_{S_1, S_2} [\bar{q}^{AB}_{S_1}, \bar{q}^{AD}_{S_2}]$ is not convergent so (\ref{qbar invariance}) is not applicable.

So far, the Hall conductance has been connected to adiabatic curvature. We now show that the definition above can also be related to charge transport. We start with a formal calculation (which, to be clear, is wrong!). By differentiating under the integral,
$$
(\beta_{2 \pi}^{AD})^{-1}(\bar{q}^{AB}) - \bar{q}^{AB} = -\int_0^{2 \pi} (\beta_\phi^{AD})^{-1} \delta_{\bar{q}^{AD}}(\bar{q}^{AB}) d\phi,
$$
however LHS and RHS are not equal as interactions based on our definitions~(\ref{def of commutator},\ref{eq:interaction_evolution}) of manipulating interactions. Continuing with formal calculations (which ignore that sums are not absolutely convergent), we conclude that
$$
\sum_S\Big((\beta_{2 \pi}^{AD})^{-1}(\bar{q}^{AB}) - \bar{q}^{AB}\Big)_S = \int_{0}^{2 \pi} (\beta_\phi^{AD})^{-1} \Big(\sum_Si[\bar{q}^{AB}, \bar{q}^{AD}]_S\Big) d\phi,
$$
and the expectation of the RHS is $-2 \pi \kappa$ by (\ref{eq:barq_hall}). This way, we obtained a formal connection between change of charge under the action of $\beta^{AD}_{2 \pi}$ and Hall conductance.

It likely won't be any surprise to the reader that to make the calculation correctly we need to regularize the expression. There are many ways how to do that, our regularization resembles \cite{BBDFanyon} (see also Lemma \ref{inu}). To this end, for $r>0$, we decompose 
\begin{equation}\label{beta decomposition}
(\beta_{2 \pi}^{AD})^{-1} =  \gamma^{1,r} \gamma^{0,r},
\end{equation}
where $\gamma$ are automorphisms such that
\begin{enumerate}
\item $\gamma^{0,r}$ (resp. $\gamma^{1,r}$) is generated by TDI $g_{0,r}$ (resp. $g_{1,r}$) anchored in $\partial(AD) \cap \{x_2 \leq r\}$ (resp. $\partial(AD) \cap \{x_2 \geq r\}$), moreover $(g_{1,r})_S = 0$ unless $S \subset AB$,
\item there exists function $f \in \mathcal{F}$ and a constant $C$ such that $\| g_{j,r}\|_f \leq C$ holds for $j=0,1$ and all $r \geq 0$, 
\item the TDIs are charge conserving, i.e. $[(g_{j,r})_S, Q_S] = 0$ for $j=0,1$, $r \geq 0$ and all finite $S$.
\end{enumerate}
For conceptual clarity, the existence of this decomposition is assumed here, with a choice of $\gamma^{0,r}$ being given explicitly when the lemma will be used in the proof of Theorem~\ref{thm:anyon}.
\begin{lem}
\label{lem:J_0_reg}
Let 
\begin{equation*}
J_0 = \int_{0}^{2\pi}(\beta^{AD}_\phi)^{-1} \Big(i\sum_{S} [\bar{q}^{AB}, \bar{q}^{AD}]_S\Big) d\phi.
\end{equation*}
Then
$$
\omega(J_0) = -2 \pi \omega(J),
$$
and 
$$
\lim_{r \to \infty} \lim_{N \to \infty} \gamma^{0,r} (\bar{Q}_{(AB)_N}) - \bar{Q}_{(AB)_N}  = J_0,
$$
where the limits are in the uniform topology of the C*-algebra.
\end{lem}
\begin{proof}
Since $[\bar{q}^{AB}, \bar{q}^{AD}]$ is summable, the equality $\omega(J_0) = -2 \pi \omega(J)$ follows immediately from~(\ref{eq:barq_hall}) and the invariance of $\omega$ under the action of~$\beta_s^{AD}$.

It remains to prove the last statement. We split the limit into two parts using $\bar{Q}_{(AB)_N} = Q_{(AB)_N} - K_{(AB)_N}$. We first consider the charge contribution. We fix $r>0$, and we are going to show that the limit of $\gamma^{0,r} (Q_{(AB)_N}) - Q_{(AB)_N}$ as $N\to\infty$ exists. For $M > N$, 
$$
{Q}_{(AB)_M} - {Q}_{(AB)_N} = \sum_{x \in (AB)_M \setminus (AB)_N} q_x.
$$ 
Using Lemma~\ref{lem:anchored_auto} we have for $|x| \gg r$,
$$
\|\gamma^{0,r}(q_x) - q_x\| \leq f(|x|/2), 
$$
and 
$$
\| (\gamma^{0,r} (Q_{(AB)_M}) - Q_{(AB)_M}) - (\gamma^{0,r} (Q_{(AB)_N}) - Q_{(AB)_N}) \| \leq \sum_{|x| \geq N} f(|x|/2).
$$
Since $f(|x|/2)$ is summable, the sum is going to zero as $N \to \infty$. Hence the sequence is Cauchy and therefore it has a limit.

The decomposition~(\ref{beta decomposition}) yields that
\begin{align*}
\gamma^{0,r} ({Q}_{(AB)_N}) - {Q}_{(AB)_N} 
      &=   (\gamma^{1,r})^{-1} ( (\beta^{AD}_{2 \pi})^{-1}({Q}_{(AB)_N}) - {Q}_{(AB)_N}) \\ &\quad+ (\gamma^{1,r})^{-1}({Q}_{(AB)_N}) - {Q}_{(AB)_N}.
\end{align*}
Since ${Q}_{(AB)_N}$ is a bonafide element of the algebra, we can differentiate under the integral sign to get
$$
(\beta^{AD}_{2 \pi})^{-1}({Q}_{(AB)_N}) - {Q}_{(AB)_N} =- \int_{0}^{2 \pi} (\beta^{AD}_\phi)^{-1}  \delta_{{\bar{q}}^{AD}}({Q}_{(AB)_N}) d\phi.
$$
Now
\begin{align*}
\delta_{\bar{q}^{AD}}(Q_{(AB)_N}) &= - \sum_{S} i[q^{(AB)_N}, \bar{q}^{AD}]_S \\
						     & =   - \sum_{S} i[q^{AB}, \bar{q}^{AD}]_S + \sum_S i [q^{(AB)_N^c}, \bar{q}^{AD}]_S,
\end{align*}
where the convergence of these sums was established in the paragraphs preceding the lemma. Hence,
\begin{equation*}
\gamma^{0,r} ({Q}_{(AB)_N}) - {Q}_{(AB)_N}  = (\gamma^{1,r})^{-1} \bigg( \int_{0}^{2 \pi} (\beta^{AD}_\phi)^{-1} \big(\sum_{S} i[q^{AB}, \bar{q}^{AD}]_S\Big) d\phi\bigg) + J_N,
\end{equation*}
where
$$
J_N = - \int_{0}^{2 \pi} (\beta^{AD}_s)^{-1} \sum_S i [q^{(AB)_N^c}, \bar{q}^{AD}]_S ds + (\gamma^{1,r})^{-1}({Q}_{(AB)_N}) - {Q}_{(AB)_N}.
$$
The automorphism $(\gamma^{1,r})^{-1}$ is generated by a TDI, let's call it $g$, that is charge conserving and supported in $AB$. Then we can write the last term as 
$$
(\gamma^{1,r})^{-1}({Q}_{(AB)_N}) - {Q}_{(AB)_N} = \sum_{S : S \cap (AB)_N^c \neq \emptyset} \int_0^1 \tau^g_s( i[(g_s)_S, {Q}_{(AB)_N}]) ds,
$$
which gives a decomposition $J_N = \sum_S (j_N)_S$ with $j_N$ anchored in $(AB)_N^c$. We established above that $J_N$ has a limit, and since it is anchored on the complement of a square that eventually covers all of $\bbZ^2$, the limit is a multiple of the identity. But $J_N$ is traceless for all $N$ and hence the limit is zero. In conclusion, we obtained
$$
\lim_{N \to \infty} \gamma^{0,r} (Q_{(AB)_N}) - Q_{(AB)_N} = (\gamma^{1,r})^{-1} \bigg( \int_{0}^{2 \pi} (\beta^{AD}_\phi)^{-1} \big(\sum_{S} i[q^{AB}, \bar{q}^{AD}]_S \big)d\phi\bigg).
$$
As $(\gamma^{1,r})^{-1}(A) \to A$ for all $A \in \caA$ we get
\begin{equation}\label{Q part}
\lim_{r \to \infty} \lim_{N \to \infty} \gamma^{0,r} (Q_{(AB)_N}) - Q_{(AB)_N} =  \int_{0}^{2 \pi} (\beta^{AD}_\phi)^{-1} \big(\sum_{S} i[q^{AB}, \bar{q}^{AD}]_S\big) d\phi.
\end{equation}

Regarding the second part associated with $K_{(AB)_N}$, Lemma~\ref{lem:anchored_auto} gives that the contribution of $(k^{(AB)_N})_S$ to $
 \gamma^{0,r} (K_{(AB)_N}) - K_{(AB)_N}$ decays with the distance of $S$ from the origin. This means that we can directly take the limit to get
$$
\lim_{N \to \infty}(\gamma^{0,r} (K_{(AB)_N}) - K_{(AB)_N}) = \sum_S \gamma^{0,r} (k^{AB}_S) - k^{AB}_S.
$$
Using the same lemma, we have that the sum on the RHS is uniformly convergent in $r$ (we assumed that TDIs $g^{0,r}$ are uniformly bounded) so we can also take the limit in $r$ to get
$$
\lim_{r \to \infty} \lim_{N \to \infty}(\gamma^{0,r} (K_{(AB)_N}) - K_{(AB)_N}) = \sum_S (\beta^{AD}_{2 \pi})^{-1} (k^{AB}_S) - k^{AB}_S.
$$
Finally, we can now differentiate term by term under the integral to get
\begin{align}\label{K part}
\lim_{r \to \infty} \lim_{N \to \infty}(\gamma^{0,r} (K_{(AB)_N}) - K_{(AB)_N}) 
&= -\sum_S \int_0^{2 \pi} (\beta_\phi^{AD})^{-1} i [\bar{q}^{AD}, k^{AB}_S] d\phi \nonumber\\
&= - \int_0^{2 \pi} (\beta_\phi^{AD})^{-1} \big(i \sum_S[\bar{q}^{AD}, k^{AB}]_S\big) d\phi, 
\end{align}
where in the second line we used that $[\bar{q}^{AD}, k^{AB}]_S$ is absolutely summable so even though the lines are not equal for each $S$, they have the same sum. (Recall (\ref{panda}).)

Adding now (\ref{Q part},\ref{K part}) back together finishes the proof.
\end{proof}

We end this section by proving that $J_0$ commutes with the ground state. We recall an elementary lemma.
\begin{lem} Let $M \in \caA$ be such that $\omega([M,A]) =0$ for all $A \in \caA$. Then 
\label{lem:comm_elementary}
\begin{equation*}
\omega(M A) = \omega(M) \omega(A)
\end{equation*}
 for any $A\in\caA$.
\end{lem}
\begin{proof}
In the GNS representation $(\caH,\pi,\Omega)$ of $\omega$, we let $P$ be the orthogonal projection onto the space spanned by $\Omega$. Since $\omega$ is pure, $\pi$ is irreducible and so $\pi(\caA)'' = (\bbC\cdot1)' = B(\caH)$, namely $\pi(\caA)$ is weakly dense in $B(\caH)$. For $B\in B(\caH)$, let $(A_\alpha)_{\alpha}$ be a net
 in $\caA$ converging weakly to $B$. Then
\begin{align*}
\Tr([P,\pi(M)]B) &= \langle\Omega,\pi(M)B\Omega \rangle - \langle\Omega,B\pi(M)\Omega \rangle \\
&=\lim_{\alpha\to\infty}\langle\Omega,\pi(M A_\alpha)\Omega \rangle - \langle\Omega,\pi(A_\alpha M)\Omega \rangle \\
&=\lim_{\alpha\to\infty} \omega([M,A_\alpha]) = 0
\end{align*}
by the assumption. Since this is true for any $B\in B(\caH)$, we conclude that $[P,\pi(M)] = 0$, and hence
\begin{equation*}
\omega(M A) = \Tr(P\pi(M) \pi(A)) = \Tr(P\pi(M)P \pi(A)) = \omega(M)\omega(A)
\end{equation*}
because $P$ is a one-dimensional projection.
\end{proof}

For any finite $\Gamma$, the conditions of the lemma are satisfied for $\bar{Q}_\Gamma$ by (\ref{qbar finite}) and we get
\begin{equation}
\label{eq:barQ}
\omega(\bar{Q}_\Gamma A) = \omega(\bar{Q}_\Gamma) \omega(A) = \omega(Q_\Gamma) \omega(A).
\end{equation}
To get the second equality we noted that because $\Gamma$ is finite,
\begin{equation*}
K_\Gamma = \int W(t) \tau_t^h(\delta_h(Q_\Gamma))dt.
\end{equation*}
Since $\omega$ is a ground state of $\tau_t^h$, it is in particular invariant and $\omega\circ\delta_h = 0$ so that $\omega(K_\Gamma) = 0$ by the formula above. It then follows from the definition of $\bar Q_\Gamma$ that $\omega(\bar Q_\Gamma) = \omega(Q_\Gamma)$.

We note for later purposes that 
\begin{equation}\label{Qbar eigenvector}
\pi(\bar Q_\Gamma)\Omega = \omega(Q_\Gamma)\Omega.
\end{equation}
This follows immediately by applying both sides of the identity $\pi(\bar{Q}_\Gamma) P = P \pi(\bar{Q}_\Gamma)$ to $\Omega$: 
\begin{equation*}
 \pi(\bar{Q}_\Gamma)\Omega = \langle \Omega,\pi(\bar{Q}_\Gamma)\Omega\rangle \Omega .
\end{equation*}

\begin{lem}
\label{lem:Invariance}
Suppose that Assumptions~\ref{assum:1} and \ref{assum:3} hold. Then
\begin{equation}
\label{eq:J0}
\omega(J_0 A) = \omega(J_0) \omega(A)
\end{equation}
holds for all $A\in\caA$.
\end{lem}
\begin{proof}
By (\ref{eq:commderivative}), we have 
$$
\omega\Big([\sum_S [\bar{q}^{AB}, \bar{q}^{AD}]_S,A]\Big) = \omega(\delta_{\bar{q}^{AB}} \delta_{\bar{q}^{AD}}(A) - \delta_{\bar{q}^{AD}} \delta_{\bar{q}^{AB}}(A)),
$$
for any local $A\in\caA$
and the RHS is equal to zero by (\ref{qbar invariance}). Using the second part of (\ref{qbar invariance}) we then get 
$$
\omega([J_0,A]) = 0,
$$
for all $A \in \caA$ and the statement follows from Lemma~\ref{lem:comm_elementary}.
\end{proof}

\section{Construction of an object in $\mathcal{M}$ associated with the $U(1)$ symmetry}
\label{sec:anyon}
Having defined the current observable $J_0$, we now turn to the explicit construction of the representation $\rho \in \mathcal{M}$ whose existence was announced in Theorem~\ref{thm:anyon} and prove that it has statistical properties stated therein. 

First of all, we note that while the TDI $\bar q^\Gamma$ that generates $\beta^\Gamma_\phi$ is anchored in $\Gamma$, the automorphism $\beta^\Gamma_{2\pi}$ has a trivial action far away from $\partial \Gamma$ because $q_x$ have integer spectrum and $k^\Gamma$ is anchored in $\partial\Gamma$. Concretely, $\beta^{\Gamma}_{2 \pi}$ can be obtained from a TDI that is anchored in $\partial\Gamma$:
\begin{lem}\label{inu}
\label{lem:twist}
Fix $\Gamma \subset \mathbb{Z}^2$. There exists a TDI, $\tilde{k}^\Gamma$, anchored in $\partial \Gamma$ such that 
$$
\beta^{\Gamma}_{2 \pi}  = \tau_{2\pi}^{\tilde{k}^\Gamma}
$$
\end{lem}
\begin{proof}
Recall that $\alpha^\Gamma_\phi$ is the family of automorphisms associated with the charge $q\vert_\Gamma$. Since $\alpha^\Gamma_{2 \pi} = \id$, we have
$$
\beta^\Gamma_{2 \pi} = \beta^{\Gamma}_{2 \pi} \circ (\alpha^{\Gamma}_{2 \pi})^{-1} = \id + \int_0^{2 \pi} \partial_\phi\left(  \beta^\Gamma_\phi \circ(\alpha^\Gamma_\phi)^{-1}\right) d\phi.
$$
Computing the derivative, we have
$$
\partial_\phi\left(  \beta^\Gamma_\phi \circ(\alpha^\Gamma_\phi)^{-1}\right)  = \beta^\Gamma_\phi \circ \left(\delta_{\bar{q}^\Gamma} - \delta_{q^\Gamma}\right)\circ (\alpha^\Gamma_\phi)^{-1}.
$$
Using $\bar{q}^\Gamma - q^\Gamma = -k^\Gamma$ we get
$$
\partial_\phi\left(  \beta^\Gamma_\phi \circ(\alpha^\Gamma_\phi)^{-1}\right) = \left(\beta^\Gamma_\phi \circ (\alpha^\Gamma_\phi)^{-1}\right) \circ \alpha^\Gamma_\phi\circ \delta_{-k^\Gamma}\circ(\alpha^\Gamma_\phi)^{-1},
$$
and the lemma holds with TDI $\tilde{k}^\Gamma(\phi) = -\alpha^\Gamma_\phi(k^\Gamma)$. The TDI is anchored in $\partial \Gamma$ by Lemma~\ref{lem:k} and Lemma~\ref{lem:anchoring}.
\end{proof}

When an arbitrary TDI $h$ is acted upon by the $U(1)$ automorphism, and only in the case, we will make an exception to~(\ref{eq:interaction_evolution})and define
$$
(\alpha_\phi(h))_S := \alpha_\phi(h_S).
$$
This is more convenient and $\alpha$ manifestly respects anchoring because it acts on-site.

To a cone $\Lambda = \ld_{\bm a,\theta,\varphi}$ we associate the half space $\Gamma_\Lambda := \{x \in \mathbb{R}^2 : f_\theta \cdot x \geq 0\}$, where $f_\theta$ is the unit vector obtained by rotating $e_\theta$ clockwise by $90$ degrees, see Figure~\ref{fig:halfplane}. Then by Lemma~\ref{lem:twist}, $\beta^{\Gamma_\Lambda}_{2 \pi}$ is generated by TDI $\tilde{k}^{\Gamma_\Lambda}$. Let $\rho^\Lambda_s$ be the family of automorphisms generated by TDI $\tilde{k}^{\Gamma_\Lambda}|_\Lambda$. Finally we put $\rho^\Lambda := \rho^\Lambda_{2 \pi}$. 

\begin{figure}
\includegraphics[width=0.5\textwidth]{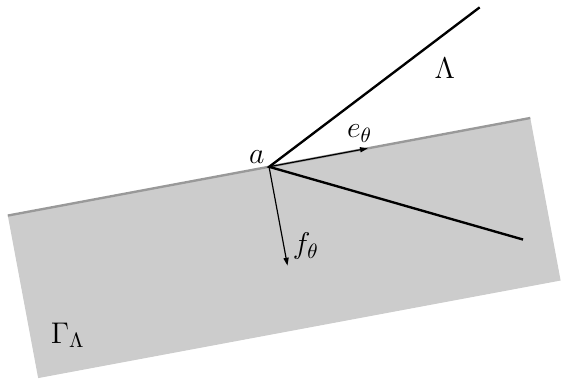}
\caption{The half plane $\Gamma_\Lambda$ associated with the cone $\Lambda$}
\label{fig:halfplane}
\end{figure}

Since $\tilde{k}^{\Gamma_\Lambda}$ is anchored on $\partial\Gamma_\Lambda$, the automorphism $\rho^\Lambda$ is generated by a TDI anchored on the axis of the cone $\Lambda$, which will be useful when we do perturbation theory. It is also possible to express $\rho^\Lambda$ using the interaction $\bar{q}$. Since restrictions commute with the on-site automorphism $\alpha_s^{\Gamma_\Lambda}$, we have that $\tilde{k}_s^{\Gamma_\Lambda}|_\Lambda = -\alpha_s^{\Lambda \cap \Gamma_\Lambda}(k^{\Gamma_\Lambda}|_\Lambda)$, and we see that
\begin{equation*}
\partial_s(\rho^\Lambda_s\circ\alpha^{\Lambda \cap \Gamma_\Lambda}_s) = (\rho^\Lambda_s\circ\alpha^{\Lambda \cap \Gamma_\Lambda}_s)\circ(-\delta_{k^{\Gamma_\Lambda}|_\Lambda} + \delta^{\Lambda \cap \Gamma_\Lambda}_{q}) = (\rho^\Lambda_s\circ\alpha^{\Lambda \cap \Gamma_\Lambda}_s)\circ \delta_{\bar q^{\Gamma_\Lambda}\vert_\Lambda}.
\end{equation*}
Since $\bar q$ is a constant interaction, this implies that $\rho_s^\Lambda = \exp\left(s\delta_{\bar q^{\Gamma_\Lambda}\vert_\Lambda}\right)\circ(\alpha_s^{\Lambda \cap \Gamma_\Lambda})^{-1}$, and in particular
\begin{equation}
\label{eq:rhobarq}
\rho^\Lambda = \exp\left(2\pi \delta_{\bar q^{\Gamma_\Lambda}\vert_\Lambda}\right).
\end{equation}
This expression will be more useful for algebraic manipulations.

The next lemma states that, for any $\Lambda$, the representation $\pi \circ \rho^{\Lambda}$ satisfies the super-selection criteria and that in addition, all these automorphisms belong to the same super-selection sector.

\begin{lem}
\label{lem:rho}
For all cones $\Lambda,\Lambda' \subset \Lambda_0$,
\begin{enumerate}
\item $\pi \circ \rho^{\Lambda} \in O_{\Lambda_0}$,
\item $ \pi \circ \rho^{\Lambda} \simeq \pi \circ \rho^{\Lambda'}$.
\end{enumerate}
Moreover there exists unitaries $V_{r,t} \in \mathcal{A}$ such that
$$
\beta_{2 \pi}^{AB} \circ \Ad [V_{r,t}] = \rho^{\Lambda_2(r)} \circ (\rho^{\Lambda_3(t)})^{-1}.
$$ 
\end{lem}
\begin{proof}
We start with the last part of the lemma. Recalling the definiton of the cones, Figure~\ref{Fig:Cones}, we see that $\Gamma_{\Lambda_3(t)} = CD$ for all $t$ and so $\rho^{\Lambda_3(t)} = \exp\left(2\pi \delta_{\bar{q}^{CD}|_{\Lambda_3(t)}}\right)$ by~(\ref{eq:rhobarq}). Using Lemma~\ref{lem:k}(iii), we have 
$$
\bar{q}^{CD}|_{\Lambda_3(t)} = q|_{\Lambda_3(t)} - \bar{q}^{AB}|_{\Lambda_3(t)},
$$
and we conclude that
$$
(\rho^{\Lambda_3(t)})^{-1} = \exp\left(2\pi \delta_{\bar{q}^{AB}|_{\Lambda_3(t)}}\right).
$$
Moreover, $\rho^{\Lambda_2(r)} = \exp\left(2\pi \delta_{\bar{q}^{AB}|_{\Lambda_2(r)}}\right)$ and since the cones $\Lambda_2(r),\Lambda_3(t)$ are disjoint, we get
$$
\rho^{\Lambda_2(r)} \circ (\rho^{\Lambda_3(t)})^{-1} = \exp\left(2\pi \delta_{\bar{q}^{AB}|_{\Lambda_2(r)\Lambda_3(t)}}\right).
$$
Equivalently, this is an automorphism generated by the TDI $\tilde{k}^{AB}|_{\Lambda_2(r) \Lambda_3(t)}$. On the other hand,
$\beta_{2 \pi}^{AB}$ is  generated by TDI $\tilde{k}^{AB}$. The claim then follows from Lemma~\ref{lem:restriction}, used for $X = \partial (AB)$ and $\Gamma = \Lambda_2(r) \Lambda_3(t)$, and Lemma~\ref{lem:pert}.

We now turn to the claim (ii). By the same reasoning that we used for the specific cones above, we get that for any non-overlapping cones $\Lambda, \Lambda'$ there exists a region $\Gamma$ and a unitary $V_{\Lambda,\Lambda'}$ such that
$$
\beta^{\Gamma}_{2 \pi} \circ \mathrm{Ad}[V_{\Lambda,\Lambda'}] = \rho^{\Lambda} \circ(\rho^{\Lambda'})^{-1},
$$
see Figure~\ref{fig:disjoint cones}. 
\begin{figure}
\includegraphics[width=0.5\textwidth]{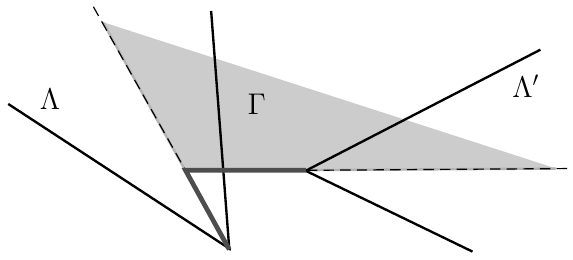}
\caption{The region $\Gamma$ corresponding to the disjoint cones $\Lambda,\Lambda'$. The unitary $V_{\Lambda,\Lambda'}$ is almost localized along the thicker grey line.}
\label{fig:disjoint cones}
\end{figure}
The invariance~(\ref{qbar invariance}) implies that $\beta^{\Gamma}_{\phi}$ is unitarily implementable in the GNS representation, namely there are unitaries $v_\phi^\Gamma$ such that $v_\phi^\Gamma \Omega = \Omega$ and
\begin{equation}\label{beta implementable}
\pi\circ \beta^{\Gamma}_{\phi} = \Ad\left[v_\phi^\Gamma\right]\circ\pi.
\end{equation}
It follows that 
\begin{equation*}
\pi\circ\rho^{\Lambda} \circ(\rho^{\Lambda'})^{-1}
= \Ad\left[v_\phi^\Gamma\right]\circ\pi\circ\mathrm{Ad}[V_{\Lambda,\Lambda'}]
= \Ad\left[v_\phi^\Gamma\pi(V_{\Lambda,\Lambda'})\right]\circ\pi
\end{equation*}
which is (ii). If $\Lambda, \Lambda'$ overlap, we find a cone $\Lambda''$, possibly ignoring the forbidden direction, that does not overlap with either. Then by the above, $\pi \circ \rho^\Lambda \simeq \pi \circ \rho^{\Lambda''} \simeq \pi \circ \rho^{\Lambda'}$, concluding the proof.

Finally, (i) holds by construction.
\end{proof}
We note that the proof provides an explicit intertwiner $V_{\rho^{\Lambda_3(t)},\Lambda_2(r)}$:
\begin{lem}
\label{lem:V}
The unitary $V_{\rho^{\Lambda_3(t)},\Lambda_2(r)} := v_{2\pi}^{AB}\pi(V_{r,t})\in\caU(\caH)$ is such that
\begin{equation*}
\left(\pi\circ\rho^{\Lambda_2(r)} (A)\right) V_{\rho^{\Lambda_3(t)},\Lambda_2(r)} = V_{\rho^{\Lambda_3(t)},\Lambda_2(r)} \left(\pi\circ\rho^{\Lambda_3(t)}(A)\right),\quad A\in\caA.
\end{equation*}
\end{lem}
In fact, more can be said. Indeed, the proof of Lemma~\ref{lem:pert} gives
\begin{equation}
\label{eq:Vrt}
V_{r,t} = \mathrm{T}\ep{i \int_0^{2\pi} G_s ds}, \quad G_s = \tau_s^{\tilde k^{AB}|_{\Lambda_2(r)\Lambda_3(t)}} \bigg(\sum_{S: S \cap (\Lambda_2(r) \Lambda_3(t))^c \neq \emptyset} \tilde{k}^{AB}_S\bigg).
\end{equation}
Since $\delta_{\tilde k^{AB}|_{\Lambda_2(r)\Lambda_3(t)}}$ acts trivially on the terms that are completely supported on the complement of $\Lambda_2(r) \Lambda_3(t)$, we have that
$$
G_s = \sum_{S: S \subset (\Lambda_2(r) \Lambda_3(t))^c} \tilde{k}^{AB}_S + \tilde{G}_s
$$
where $\tilde{G}_s$ is an observable that is almost localized at the apexes, $a_2(r), a_3(t)$, of $\Lambda_2(r)$ and $\Lambda_3(t)$. Specifically, there exists $f \in \mathcal{F}$ and a constant $C$, both independent of $r,t$, such that $\| \tilde{G}_s \|_f \leq C$ and $\|\tilde{G}_s(S)\| \leq  f(\mathrm{dist}(S, \{a_2(r), a_3(t)\})$.

\begin{lem}\label{vab}
Let $c_N=\ep{2\pi{i}\omega(Q_{AB_{N}})}$. Then
\begin{equation*}
v_{2\pi}^{AB}
=\slim_{N\to\infty} \bar{c}_N \pi\lmk \ep{2\pi i \bar{Q}_{AB_{N}}}\rmk
\end{equation*}
\end{lem}
\begin{proof}
First we note that~(\ref{Qbar eigenvector}) implies that $\pi(\ep{2 \pi i \bar{Q}_{AB_N}}) \Omega = c_N\Omega$. Then for any $A\in\caA$, 
\begin{equation*}
\bar c_N\pi\lmk\ep{2\pi{i} \bar{Q}_{AB_{N}}}\rmk\pi(A)\Omega = \pi\left(\Ad\left[\ep{2\pi{i} \bar{Q}_{AB_{N}}}\right](A)\right) \Omega
\end{equation*}
while
\begin{equation*}
v_{2\pi}^{AB}\pi(A)\Omega = \pi\left(\beta_{2\pi}^{AB} (A)\right)\Omega
\end{equation*}
by~(\ref{beta implementable}). Hence
\begin{multline*}
\bar c_N\pi\lmk\ep{2\pi{i} \bar{Q}_{AB_{N}}}\rmk\pi(A)\Omega - v_{2\pi}^{AB}\pi(A)\Omega 
=\pi\left(\Ad\left[\ep{2\pi{i} \bar{Q}_{AB_{N}}}\right](A) - \beta_{2\pi}^{AB}(A)\right)\Omega
\end{multline*}
and the claim follows from the strong limit
\begin{equation*}
\beta_{2\pi}^{AB}= \slim_{N\to\infty}\Ad\left[\ep{2\pi{i} \bar{Q}_{AB_{N}}}\right]
\end{equation*}
and the cyclicity of $\Omega$ with respect to $\pi(\caA)$.
\end{proof}
\begin{rem}\label{neko}
Note from (\ref{beta implementable}) that $v_{2\pi}^{AB}$ can be arbitrarily approximated by elements
in $\cup_{\Lambda\in\caC}\pi(\caA_{\Lambda^c})'$.
In particular, by the approximate Haag duality, $v_{2\pi}^{AB}$ belongs to $\caB$.
More precisely, for each $\varepsilon>0$, we may choose a cone $\Lambda_\varepsilon\in\caC$ and 
a unitary $U_\varepsilon\in \caA$
such that $\pi(U_\varepsilon) v_{2\pi}^{AB}\in \pi\lmk\caA_{\Lambda_\varepsilon^c}\rmk'$
and $\lV U_\varepsilon-\unit\rV<\varepsilon$.
Note also that we may further choose a sequence of unitaries $U_\varepsilon^N\in \caA$
such that $\pi\lmk U_{\varepsilon}^N\ep{2\pi{i} \bar{Q}_{AB_{N}}}\rmk \in \pi\lmk\caA_{\Lambda_\varepsilon^c}\rmk'$
and $\lV U_{\varepsilon}^N- U_{\varepsilon}\rV\to 0$, $N\to\infty$.
Hence we obtain
\begin{align}
\pi(U_\varepsilon)v_{2\pi}^{AB}=\slim_{N\to\infty} \bar{c}_N \pi( U_{\varepsilon}^N)\pi\lmk \ep{2\pi i \bar{Q}_{AB_{N}}}\rmk
\end{align}
in $\pi\lmk\caA_{\Lambda_\varepsilon^c}\rmk'$.
\end{rem}

We are now ready to prove our first main result, namely that the braiding statistics $\theta(\rho,\rho)$ is nothing but the exponential of the Hall conductance, which in turn was defined via~(\ref{eq:Hall},\ref{def:J}) as the expectation value of the adiabatic curvature.

\begin{proof}[Proof of Theorem~\ref{thm:anyon}]
We claim that $\rho := \pi \circ \rho^{\Lambda_1}$ has the stated properties. By Lemma~\ref{lem:rho}, $\rho \in \mathcal{M}$, and it is simple because $\pi$ is irreducible. So it remains to prove the braiding relation $\theta(\rho, \rho) = \ep{ 2 \pi {i} \omega(J_0)}$.

Using item (iv) in Lemma~\ref{lem:theta} and Lemma~\ref{lem:rho}, we have for any $r \geq 0$, $\theta(\rho, \rho) = \theta(\pi \circ \rho^{\Lambda_1(r)}, \rho)$. In particular,
$$
\theta(\rho, \rho) = \lim_{r \to \infty} \theta\left(\pi \circ \rho^{\Lambda_1(r)}, \rho\right).
$$
 We now pick $V_{\rho, \Lambda_3(t)}$ such that 
$\Ad  V_{\rho, \Lambda_3(t)} \circ \rho = \pi \circ \rho^{\Lambda_3(t)}$ to get 
$$
\theta\left(\pi \circ \rho^{\Lambda_1(r)}, \rho\right) = \lim_{t \to \infty} \epsilon(\pi \circ \rho^{\Lambda_1(r)}, \pi \circ \rho^{\Lambda_3(t)}).
$$
By Lemma~\ref{lem:V},
$$
\epsilon(\pi \circ \rho^{\Lambda_1(r)}, \pi \circ\rho^{\Lambda_3(t)}) = \lim_{s \to \infty}  \pi(V_{s,t}^*) (v_{2 \pi}^{AB})^* T_{\rho^{\Lambda_1(r)}} ( v_{2 \pi}^{AB} \pi(V_{s,t})).
$$
Now 
\begin{align*}
T_{\rho^{\Lambda_1(r)}} ( v_{2 \pi}^{AB} \pi(V_{s,t})) 
=T_{\rho^{\Lambda_1(r)}} ( \pi\beta_{2\pi}^{AB}(V_{s,t}) v_{2 \pi}^{AB} ) 
=\pi(\rho^{\Lambda_1(r)} \beta_{2\pi}^{AB}(V_{s,t}))T_{\rho^{\Lambda_1(r)}} (v_{2 \pi}^{AB} ) 
\end{align*}
and the explicit expression~(\ref{eq:Vrt}) implies that $\lim_{r \to \infty} \rho^{\Lambda_1(r)}\beta_{2\pi}^{AB}(V_{s,t}) 
=\beta_{2\pi}^{AB} \lmk V_{s,t}\rmk$ uniformly in $s,t$. Hence,
$$
\theta(\rho, \rho) = \lim_{r \to \infty} (v_{2 \pi}^{AB})^*T_{\rho^{\Lambda_1(r)}}(v_{2 \pi}^{AB}) = \lim_{r \to \infty} \left\langle\Omega,(v_{2 \pi}^{AB})^* T_{\rho^{\Lambda_1(r)}}(v_{2 \pi}^{AB}) \Omega\right\rangle,
$$
where we used that $\theta(\rho, \rho)$ is a scalar in the last equality. With this, the weak continuity of $T_{\rho^{\Lambda_1(r)}}$ on each $\pi\lmk\caA_{\Lambda^c}\rmk'$, $\Lambda\in\caC$, Remark \ref{neko} and Lemma~\ref{vab}, we get 
$$
\theta(\rho, \rho) = \lim_{r \to \infty} \lim_{N \to \infty} \omega( \ep{-2\pi i \bar{Q}_{(AB)_N}} \rho^{\Lambda_1(r)} ( \ep{2\pi i \bar{Q}_{(AB)_N}})).
$$
Now we consider the automorphism $\theta_r := (\beta^{AD}_{2 \pi})^{-1} \circ \rho^{\Lambda_1(r)}$ of $\caA$, which is such that
\begin{equation}\label{beta rho away from cone}
\lim_{r\to\infty}\Vert (\beta^{\Gamma_{\Lambda_1}}_{2 \pi})^{-1}\circ\rho^{\Lambda_1(r)}(A) - (\beta^{\Gamma_{\Lambda_1}}_{2 \pi})^{-1}(A)\Vert = 0
\end{equation} 
holds for all all $A\in\caA$. Using~(\ref{eq:barQ}), we have that 
\begin{equation}
\label{eq:theta using theta t}
\theta(\rho, \rho) = \lim_{r \to \infty} \lim_{N \to \infty} \omega( \ep{-2\pi i \bar{Q}_{(AB)_N}} \theta_r ( \ep{2\pi i \bar{Q}_{(AB)_N}})).
\end{equation} 

Let now $v_{r,N}(\phi) = \ep{- {i}\phi \bar{Q}_{(AB)_N}} \theta_r ( \ep{ {i} \phi\bar{Q}_{(AB)_N}})$ and $v(\phi) = \lim_{r \to \infty} \lim_{N \to\infty}v_{r,N}(\phi)$ uniformly in $\phi$, which are such that $v_{r,N}(0) = 1$ and therefore $v(0) = 1$. Note that $\gamma^{0,r} = \theta_r$ is a possible choice in the decomposition~(\ref{beta decomposition}), so by Lemma~\ref{lem:J_0_reg} we have that 
\begin{equation}
\label{eq:theta_to_Hall}
\lim_{r \to \infty} \lim_{N \to\infty} (\theta_r(\bar{Q}_{(AB)_N}) - \bar{Q}_{(AB)_N}) = J_0,
\end{equation} 
and since
$$
\partial_\phi(\ep{-{i}\phi \bar{Q}_{(AB)_N}} \theta_r ( \ep{{i} \phi \bar{Q}_{(AB)_N}})) = {i} \ep{- {i} \phi \bar{Q}_{(AB)_N}} (\theta_r(\bar{Q}_{(AB)_N}) - \bar{Q}_{(AB)_N}) \theta_r ( \ep{{i} \phi \bar{Q}_{(AB)_N}}).
$$
we conclude from~(\ref{eq:theta_to_Hall}
) that
$$
\partial_\phi v_\phi = {i} \beta_{-\phi}^{AB}\lmk J_0\rmk v_\phi.
$$
the proof of 
Lemma~\ref{lem:Invariance} now implies that $\partial_\phi \omega(v_\phi) = {i} \omega(J_0) \omega(v_\phi)$ and therefore
\begin{equation*}
\theta(\rho, \rho) = \omega(v_{2 \pi}) = e^{2 \pi {i} \omega(J_0)},
\end{equation*}
which is the equation that we aimed to prove.
\end{proof}

\section{Quantization of Hall conductance}
\label{sec:quantization}
We finally prove Theorem~\ref{thm:hall}. We start with a Lemma that is a corollary of Lemma~\ref{lem:theta}.

\begin{lem}
Suppose that $\rho \simeq \pi$, then $\theta(\rho, \sigma) = 1$ for any $\sigma \in O_{\Lambda_0}$.
\end{lem}
\begin{proof}
By point (iv) in Lemma~\ref{lem:theta} we have 
$$
\theta(\rho, \sigma) = \theta(\pi, \sigma).
$$
The right hand side is manifestly equal to the unity.
\end{proof}
Now we are ready for the proof of the theorem.
\begin{proof}[Proof of Theorem~\ref{thm:hall}]
Consider the object $\rho$ constructed in Theorem~\ref{thm:anyon}. For any $n\in \mathbb{N}$, $\rho^{\otimes n}$ is an irreducible object in $\mathcal{M}$. By the assumption that there are finite number $p'$ of simple objects, there exists $p \leq p'$ such that $\rho^{\otimes p} \simeq \pi$. By the previous lemma we then have 
$$
\theta(\rho^{\otimes p}, \rho) = 1.
$$
On the other hand by Theorem~\ref{thm:anyon} and Lemma~\ref{lem:theta}(v),
$$
\theta(\rho^{\otimes p}, \rho) = \theta(\rho, \rho)^p = e^{2 \pi i  \omega(J_0) p}.
$$
The two equations and the definition~(\ref{eq:Hall}) then imply the stated quantization of the Hall conductance $\kappa$.
\end{proof}

\appendix

\section{Manipulating interactions}
\label{app:interactions}
We use one of the standard setups to manipulate interactions. We follow \cite{BDFJ}.
We consider the $l^\infty$-norm on $\bbZ^2$.
For $x\in \bbZ^2$ and $r>0$, $B_r(x)$ indicates the ball in $\bbZ^2$ centered at $x$ with radius $r$
with respect to this norm.
The diameter of a subset $S\subset \bbZ^2$ with respect to this norm is denoted by
$\mathrm{diam}(S)$.

\subsection{Interactions}
Let $\mathcal{F}$ be a class of strictly positive, non-decreasing functions $f : \mathbb{N}^+ \to \mathbb{R}^+$ that decay faster than any power, i.e. $\lim_{r \to \infty} f(r) r^p = 0$ for all $p >0$. An interaction $h : S \Subset \mathbb{Z}^2 \to h_S \in \mathcal{A}_S$ is a map associating a finite subset of $\mathbb{Z}^2$ with an operator in $\mathcal{A}_S$. We will only consider interactions for which
$$
\| h \|_f = \sup_{x \in \mathbb{Z}_2} \sum_{S \ni x} \frac{\|h_S\|}{f(1 + \mathrm{diam}(S))}
$$
is finite for some $f \in \mathcal{F}$, we will call the set of all such interactions by $\mathcal{J}$.  We will denote interactions with lower case letters and use upper case letter for their quadratures. For a set $X$ we  put
$$
H_X = \sum_{S : S \cap X \neq \emptyset} h_S,
$$
whenever the sum converges in the norm topology. The sum, in particular, converges for $X$ finite, provided $h \in \mathcal{J}$. If the sum exists for $X= \mathbb{Z}^2$, we call the interaction summable. A derivation $\delta^h_X$, associated with $h$ and a region $X$, acts as 
$$
\delta^h_X(A) = \sum_{S : S \cap X \neq \emptyset} i[h_S, A], \quad A \in \mathcal{A}_{\mathrm loc}.
$$
For $X = \mathbb{Z}^2$, we omit $X$ and write $\delta^h$. If $\delta^h$ is inner, we will call the interaction $h$ an inner interaction. Summable interactions are inner, but there are inner interactions that are not summable. If $H_X$ exists then $\delta^h_X(A) = i[H_X, A]$.

\subsection{Anchored interactions} We say that an interaction $h$ is anchored in a set $X \subset \mathbb{Z}^2$ if $h \in \mathcal{J}$ and $S \cap X = \emptyset$ implies that $h_S = 0$. If an interaction is anchored in a finite region $X$ then $H_X = H_{\mathbb{Z}^2}$ exists.

Anchoring can be readily connected to more standard forms of locality.
\begin{lem}
\label{lem:anchored_commutator}
Suppose that $h$ is anchored in $X$ and $A \in \mathcal{A}_Y$. Then
$$
\| \delta_{h}(A) \| \leq 2 f(\mathrm{dist}(X,Y)) |Y| \|h\|_f \|A\|,
$$
holds for any $f \in \mathcal{F}$.
\end{lem}
\begin{proof}
We have
$$
\delta_{h}(A) = \sum_{\substack{S : S \cap X \neq \emptyset, \\ S \cap Y \neq \emptyset}} i[h_S,A].
$$
So we get
\begin{align*}
\| \delta_{h}(A) \| &\leq 2\sum_{y \in Y} \sum_{\substack{S : y \in S, \\ S \cap X \neq \emptyset}} \| h_S\| \|A\| \\
			 & \leq 2\sum_{y \in Y} \sum_{\substack{S : y \in S, \\ S \cap X \neq \emptyset}}  \frac{\| h_S\|}{ f(1 + \mathrm{diam}(S))} f(1 + \mathrm{diam}(S)) \|A\|  \\
			 & \leq 2|Y| \|h\|_f f(1 + \mathrm{dist}(X,Y)) \|A\|,
\end{align*}
which is what we were supposed to prove. We note that the RHS might be infinite in which case the inequality is trivial.
\end{proof}

\subsection{Commutators} For interactions $h,h'$ we define their commutator $[h,h']$ as 
\begin{equation}\label{def of commutator}
[h,h']_S = \sum_{\substack{S_1, S_2: S_1 \cup S_2 = S, \\ S_1 \cap S_2 \neq \emptyset}} [h_{S_1}, h'_{S_2}].
\end{equation}
If $h, h' \in \mathcal{J}$ then $[h,h'] \in \mathcal{J}$, and 
\begin{equation}
\label{eq:commderivative}
-i \delta_{[h,h']} = \delta_h \delta_{h'} - \delta_{h'} \delta_h.
\end{equation}
Furthermore, if $h$ is anchored in $X$ then $[h,h']$ is anchored in $X$.

It would be convenient if $[h,h']$ was anchored in the intersection of the anchors of $h, h'$. Alas, this is not the case. As a partial substitute, we will use the following criteria to decide if a commutator is inner.
\begin{lem}
\label{lem:inner_commutator}
Let $h$ (resp. $h'$) be interactions anchored in $X$ (resp. $X'$). Assume that for all $f \in \mathcal{F}$,
$$
\sum_{x \in X} \sum_{x' \in X'} f( |x - x'|) < \infty.
$$
Then $[h,h']$ is summable.
\end{lem}
\noindent Note that the assumption above is to be understood as a constraint on the sets and not on the family $\caF$. It is satisfied in particular whenever $X,X'$ are two non-parallel strips of finite width.
\begin{proof}
By the definition~(\ref{def of commutator}) of the commutator of interactions, it suffices to show that 
\begin{equation*}
\caS = \sum_{S\cap X \neq\emptyset}\sum_{\substack{S'\cap X' \neq\emptyset \\ S\cap S'\neq\emptyset}}\Vert h_S\Vert \Vert h'_{S'}\Vert 
\end{equation*}
is convergent. Let $\tilde f$ be such that $\Vert h\Vert_{\tilde f}<\infty$, $\Vert h'\Vert_{\tilde f}<\infty$ and let
$$
g(r) := \max_{r_1 + r_2 = r+1} {\tilde f}(r_1) {\tilde f}(r_2) .
$$
Then $ g \in \mathcal{F}$ and we write
$$
\caS \leq \sum_{\substack{x \in X \\ x \in X'}} \sum_{\substack{S: x \in S, x' \in S' \\  S\cap S'\neq\emptyset}} \frac{\Vert h_S\Vert}{{\tilde f}(1 + \mathrm{diam}(S))} \frac{\Vert h'_{S'}\Vert}{{\tilde f}(1+ \mathrm{diam}(S'))} g(2 + \mathrm{diam}(S) + \mathrm{diam}(S')).
$$
By the geometry of $S, S'$ in the above sum, $\mathrm{diam}(S) + \mathrm{diam}(S') \geq |x- x'|$ so we get
$$
\caS \leq \|h\|_{\tilde f} \|h'\|_{\tilde f} \sum_{x \in X} \sum_{x \in X'} g(|x - x'|),
$$
which is finite by assumption.
\end{proof}

\subsection{Time-evolution}
For $A \in \mathcal{A}$ and a site $x \in \mathbb{Z}^2$, we define a decomposition of $A$
$$
A = \sum_{n=0}^\infty A_{x,n},
$$
where
$$
A_{x,n} := \mathbb{E}_{B_n(x)}[A] - \mathbb{E}_{B_{n-1}(x)}[A], 
$$
for $n \geq 1$ and $A_{x,0} := \mathbb{E}_{B_0(x)}[A]$.
For an automorphism $\beta$, and an interaction $h$ anchored in $X$, we define the time evolved interaction as
\begin{equation}
\label{eq:interaction_evolution}
\beta(h)_{B_k(x)} := \sum_{S : x \in S \cap X} \frac{1}{|S \cap X|} \beta(h_S)_{x,k},
\end{equation}
for $x \in X$ and $k \geq 0$. We define $\beta(h)_S = 0$ for any other $S$. 

For an interaction $h$, we denote $\tau^h_s$ the group of automorphisms generated by $\delta^h$. We will repeatedly use 
\begin{lem}\label{lem:anchoring}
Suppose that $h \in \mathcal{J}$ and that $h'$ is anchored in $X$. Then $\tau^{h}_s(h')$ is anchored in $X$ for all $s \geq 0$.
\end{lem}
The proof is in  \cite[Lemma~5.2.]{BDFJ}. 

\subsection{Quadratures}
For a series of interactions $h_j$, we define
$$
\left( \sum_j h_j \right)_S := \sum_j (h_j)_S,
$$
provided the sum exists in norm sense. Likewise, provided that the integral on the RHS exists in the Bochner sense, we put
$$
\left( \int h_t w(t) dt \right)_S =  \int (h_t)_S w(t) dt,
$$
for a family of interactions $h_t$ and weight function $w(t)$.

\subsection{Restrictions}
\label{app:restrictions}
For an interaction $h$ and a region $\Gamma$ we define 
$$
(h |_\Gamma)_S = \begin{cases}
					0&S \cap \Gamma^c \neq \emptyset  \\
					h_S& S \subset \Gamma .
				\end{cases}
$$
The automorphisms, $\tau^{h|_\Gamma}_s$, associated with $h |_\Gamma$ act strictly in $\Gamma$.

\subsection{Time dependent interactions}
Time dependent interaction (TDI) is a map $s \in I \subset \mathbb{R} \to h_s \in \mathcal{J}$,
with $I$ an interval of $\bbR$. Furthermore we require that 
\begin{enumerate}
\item the map $s \in I \to (h_s)_S$ is continuous for all $S \subset \mathbb{Z}^2$,
\item there exist $f \in \mathcal{F}$ such that 
$$
\sup_{s \in I} \| h_s \|_f < \infty.
$$
\end{enumerate}
Operations on interactions extend point-wise  to TDIs. The role of TDI's is to generate time evolution, to a TDI $h$ we associate a family of automorphisms $\tau^h_s$ that satisfies the equation 
$$
\partial_s \tau^h_s(A) = \tau^h_s(\delta_{h_s}(A)),\quad A\in {\caA_{\mathrm{loc}}}.
$$

Anchored interactions generate an automorphism that acts trivially far away from the anchoring region. This is quantified by the following lemma.
\begin{lem}
\label{lem:anchored_auto}
Let $h$ be a TDI anchored in $X$, $\tau^h_s$ the associated automorphism, and $A \in \mathcal{A}_Y$,
$$
\| \tau^h_1(A) - A \| \leq 2|Y| \sup_{s \in [0,1]} \| h_s\|_f f(\mathrm{dist}(X,Y)) \|A\|,
$$
holds for any $f \in \mathcal{F}$.
\end{lem}
\begin{proof}
By differentiating under integral we get
$$
\tau_1^h(A) - A = \int_0^1 \tau^h_s( \delta_{h_s}(A)) ds,
$$
and the statement follows from Lemma~\ref{lem:anchored_commutator}.
\end{proof}

\subsection{Perturbation theory}
\begin{lem}
\label{lem:pert}
Let $h, h'$ be two TDIs, and $\tau^h_s$, $\tau^{h'}_s$ the associated automorphisms. Suppose that $h- h'$ is inner, i.e. there exists a family $D_s \in \mathcal{A}$ such that
\begin{equation}
\label{eq:pert1}
\delta_{h_s}(A) - \delta_{h'_s}(A) = i[D_s,A],
\end{equation}
holds for all $A \in \mathcal{A}$. Then there exists a unitary $V_s \in \mathcal{A}$ such that 
$$
V_s \tau_h^s(A) = \tau_{h'}^s(A) V_s. 
$$
holds for all $A \in \mathcal{A}$.
\end{lem}
\begin{proof}
We have 
$$
\partial_s (\tau_h^s \circ(\tau_{h'}^s)^{-1}(A)) = \tau_h^s(\delta_{h_s} - \delta_{h'_s}) (\tau_{h'}^s)^{-1}(A) = \tau_h^s \circ(\tau_{h'}^s)^{-1}[i\tau_{h'}^s(D_s),A],
$$
by the assumption. In other words, the family of automorphisms $\tau_h^s \circ (\tau_{h'}^s)^{-1}$ is generated by the family of self-adjoint elements $\tau_{h'}^s(D_s)\in \caA$. It follows immediately that $\tau_h^s \circ (\tau_{h'}^s)^{-1} = \Ad[V_s^*]$ where $V_s$ is the time-ordered exponential of $\tau_{h'}^t(D_t)$.
\end{proof}
The lemma will be mainly used in the context of localizing interactions.

\begin{lem}
\label{lem:restriction}
Let $h$ be a TDI anchored in a region $X$. Suppose that $\Gamma \subset \mathbb{Z}^2$ is a region such that there exists constants $C_1, C_2$ so that 
$$
\mathrm{dist}(\Gamma^c, X \cap B_{0,n}^c) \geq C_1 + C_2 n
$$
holds for all integers $n$ with $B_{0,n}:=B_0(n)$. 
Then the TDI $h - h|_{\Gamma}$ is inner. 
\end{lem}
\begin{proof}
Since, by definition,
$$
(h - h\vert_{\Gamma})_S = \begin{cases}
						0 &	S \subset \Gamma,  \\
						 h_S & \mbox{otherwise},
					\end{cases}
$$
we get
$$
\sum_{S} \| (h - h\vert_{\Gamma})_S \| = \sum_{S \cap \Gamma^c \neq \emptyset} \|h_S\|. 
$$
Since $h$ is anchored in $X$ we can add a condition $S \cap X \neq \emptyset$ to the last sum. Then we bound it as 
$$
\sum_{\substack{S \cap \Gamma^c \neq \emptyset \\ S  \cap X \neq \emptyset}} \|h_S\| \leq \sum_{n=0}^\infty \sum_{x \in B_n \cap X} \sum_{\substack{S \ni x \\ S \cap \Gamma^c \neq \emptyset \\ S  \cap X\cap B_{n-1}^c \neq \emptyset}} \|h_S\|.
$$
Any set $S$ in the last sum is such that includes points in both $\Gamma^c$ and  $X \cap B_{0,n-1}^c$. The diameter of such set is bigger than $C_1 + C_2(n-1)$ by assumption. If $f$ is such that $\|h\|_f < \infty$, then 
\begin{align*}
\sum_{S} \| (h - h\vert_{\Gamma})_S \| &\leq \sum_{n=0}^\infty \sum_{x \in B_n \cap X} \sum_{\substack{S \ni x \\ S \cap \Gamma^c \neq \emptyset \\ S  \cap B_{n-1}^c \neq \emptyset}} \frac{\|h_S\|}{f(1 + \mathrm{diam(S))}} f(1 + C_1 + C_2(n-1)) \\ 
&\leq \|h\|_f \sum_{n=0}^\infty  (2n+1)^2 f(1 + C_1 + C_2(n-1)), 
\end{align*}
and the series is converent since $f$ decays faster than any inverse power.
\end{proof}

\section{Braiding statistics associated with winding}
\label{app:theta}

The goal of this appendix is to finish the proof of Lemma~\ref{lem:theta}. Throughout the appendix we assume that  assumptions \ref{assum:1}, \ref{assum:2} hold.

A technical tool that we will use is a Lemma that follows from approximate Haag duality, see \cite[Lemma 2.5]{MTC}. We will use the notation $(\ld_{\bm a,\theta,\varphi})_\epsilon:=\ld_{\bm a,\theta,\varphi+\epsilon}$.

\begin{lem}\label{usingApproxHaagDuality1}
    Let $\varepsilon > 0$, and $\delta > 0$, and let $\Lambda$ be a cone such that $\abs{\arg \Lambda} + 4 \varepsilon < 2 \pi$. 
    Let $A \in \pi(\caA_{\Lambda^c})'$. Then, under the assumption of approximate Haag duality, for all $r>R_{\abs{\arg  \Lambda},\varepsilon}$ there exists $A'_r \in \pi(\caA_{(\Lambda(-r))_{\varepsilon+\delta}})''$ such that $\norm{A - A'_r} \le 2 f_{\abs{\arg \Lambda},\varepsilon,\delta}(r) \norm{A}$.  Here $f_{\cdot}(r)$ is a decreasing function that vanishes in the limit $r \to \infty$. 
    
    Specifically, there exists a unitary, $\Tilde{U}_{r}$, depending on $\Lambda, \varepsilon, \delta$, such that  $A'_r = \Ad(\Tilde{U}_{r})A$ satisfies these conditions.
\end{lem}

\begin{lem}\label{VSigmaLambdaThreeTwoCommutatorPiCalALimit}
 Let $\sigma \in O_{\Lambda_0}$. Then $$\lim\limits_{s,t\to\infty}\norm{[V_{\sigma,\Lambda_3(s)}V_{\sigma,\Lambda_2(t)}^*,A]}=0$$ holds for all $A \in \pi(\caA)$.
\end{lem} 
\begin{proof}
    First, suppose that $A \in \pi(\caA_{loc})$. Then, there is some finite set on which $A$ is supported. So, as $\bigcup_{n \in \bbN}\Lambda_1(-n) = \bbR^2$, there exists some $n \in \bbN$ such that $A$ is supported in $\ld_1(-n)$, i.e., $A \in \pi(\caA_{\ld_1(-n)}) \subseteq \pi(\caA_{\ld_1(-n)})''$. For $s,t > n \tan(\frac{\pi}{8})$, $\ld_1(-n) \subset (\ld_3(s)\cup\ld_2(t))^c$, and so $\pi(\caA_{\ld_1(-n)})'' \subset (\pi(\caA_{(\ld_3(s)\cup\ld_2(t))^c})')'$, and therefore $A \in \pi(\caA_{\ld_1(-n)})'' \subset (\pi(\caA_{(\ld_3(s)\cup\ld_2(t))^c})')'$. By  \cite[Lemma~2.2]{MTC}, $V_{\sigma,\Lambda_3(s)}V_{\sigma,\Lambda_2(t)}^* \in \pi(\caA_{(\ld_3(s)\cup\ld_2(t))^c})'$, and so $[A,V_{\sigma,\Lambda_3(s)}V_{\sigma,\Lambda_2(t)}^*]=0$.
    
    We conclude that for all  $A \in \pi(\caA_{\mathrm{loc}})$, $\lim\limits_{s,t\to\infty} \norm{[V_{\sigma,\Lambda_3(s)}V_{\sigma,\Lambda_2(t)}^*,A]}=0$.
For $A \in \pi(\caA)$, the statement follows by density of $\pi(\caA_{loc})$ in $\pi(\caA)$.
\end{proof}

\begin{lem}\label{VSigmaLambdaThreeTwoCommutatorLambda1CompPrimeLimit}
 Let $\sigma \in O_{\Lambda_0}$. Then $$\lim\limits_{s,t\to\infty} \norm{[V_{\sigma,\Lambda_3(s)}V_{\sigma,\Lambda_2(t)}^*,A]}=0$$ holds for all $A \in \pi(\caA_{\Lambda_1^c})'$.
\end{lem}
\begin{proof}
    Let $A \in \pi(\caA_{\Lambda_1^c})'$.
    Pick $\varepsilon > 0$ and $\delta > 0$ such that $\abs{\arg \ld_1}+4\varepsilon < 2\pi$. For concreteness, pick $\varepsilon=\delta=10^{-3}$. Then, by Lemma \ref{usingApproxHaagDuality1}, for all $r > R_{\abs{\arg \ld_1},\varepsilon}$ there exists $A'_r \in \pi(\caA_{(\ld_1(-r))_{\varepsilon+\delta}})''$ such that $\norm{A-A'_r}\le 2 f_{\abs{\arg \ld_1},\varepsilon,\delta}(r)$.
    If $(\ld_1(-r))_{\varepsilon+\delta} \subset (\ld_3(s)\cup\ld_2(t))^c$, then $A'_r \in \pi(\caA_{(\ld_1(-r))_{\varepsilon+\delta}})'' \subseteq (\pi(\caA_{(\ld_3(s)\cup\ld_2(t))^c})')'$.
    By \cite[Lemma~2.2]{MTC}, $V_{\sigma,\Lambda_3(s)}V_{\sigma,\Lambda_2(t)}^* \in \pi(\caA_{(\ld_3(s)\cup\ld_2(t))^c})'$. Therefore, $A'_r$ commutes with $V_{\sigma,\Lambda_3(s)}V_{\sigma,\Lambda_2(t)}^*$.
    So, whenever $(\ld_1(-r))_{\varepsilon+\delta} \subset (\ld_3(s)\cup\ld_2(t))^c$,
    \begin{align*}
    \norm{[V_{\sigma,\Lambda_3(s)}V_{\sigma,\Lambda_2(t)}^*, A]}&=\norm{[V_{\sigma,\Lambda_3(s)}V_{\sigma,\Lambda_2(t)}^*, A'_r + (A-A'_r)]}\\
    &= \norm{[V_{\sigma,\Lambda_3(s)}V_{\sigma,\Lambda_2(t)}^*, A'_r]+[V_{\sigma,\Lambda_3(s)}V_{\sigma,\Lambda_2(t)}^*, (A-A'_r)]}\\
    &\le \norm{[V_{\sigma,\Lambda_3(s)}V_{\sigma,\Lambda_2(t)}^*, A'_r]}+\norm{[V_{\sigma,\Lambda_3(s)}V_{\sigma,\Lambda_2(t)}^*, (A-A'_r)]}\\
    &\le 0 + 2 \norm{V_{\sigma,\Lambda_3(s)}V_{\sigma,\Lambda_2(t)}^*} \norm{A-A'_r}\\
    &\le 4 f_{\abs{\arg \ld_1},\varepsilon,\delta}(r).
    \end{align*}
    Therefore, pick $r=\max(R_{\abs{\arg \ld_1}},\cot(\frac{\pi}{8} + \varepsilon+\delta) \cdot \min(t,s)  - 1)$, so that for sufficiently large $s,t$, $(\ld_1(-r))_{\varepsilon+\delta} \subset (\ld_3(s)\cup\ld_2(t))^c$ and also so that $r \to \infty$ as $\min(s,t) \to \infty$, so that this upper bound of $4 f_{\abs{\arg \ld_1},\varepsilon,\delta}(r)$ on $\norm{[V_{\sigma,\Lambda_3(s)}V_{\sigma,\Lambda_2(t)}^*, A]}$ goes to $0$ as $t,s \to \infty$. So, $\lim\limits_{s\to\infty}\lim\limits_{t\to\infty} \norm{[V_{\sigma,\Lambda_3(s)}V_{\sigma,\Lambda_2(t)}^*,A]}=0$.
\end{proof}

\begin{lem}\label{VStarVInSigmaPrime}
    Let $\sigma \in O_{\Lambda_0}$. For $i=a,b$ let $\Lambda_i$ be a cone, and $V_{\sigma,\Lambda_i} \in \caV_{\sigma,\Lambda_i}$. 
    Then, $V_{\sigma,\Lambda_b}^*V_{\sigma,\Lambda_a} \in \sigma(\caA_{(\Lambda_a \cup \Lambda_b)^c})'$.
\end{lem}
\begin{proof}
    This proof is very similar to that of  \cite[Lemma~2.2]{MTC}.
    Let $A \in \caA_{(\Lambda_a \cup \Lambda_b)^c}=\caA_{\Lambda_a^c}\cap\caA_{\Lambda_b^c}$. Then
    $$\Ad(V_{\sigma,\Lambda_b}^*V_{\sigma,\Lambda_a})\circ\sigma(A)=\Ad(V_{\sigma,\Lambda_b}^*)\circ\pi(A)=\sigma(A).$$
    So, for all $  A \in \caA_{(\Lambda_a \cup \Lambda_b)^c}$, $[V_{\sigma,\Lambda_b}^*V_{\sigma,\Lambda_a},\sigma(A)]=0$, i.e. $V_{\sigma,\Lambda_b}^*V_{\sigma,\Lambda_a} \in \sigma(\caA_{(\Lambda_a \cup \Lambda_b)^c})'$.
\end{proof}

\begin{lem}\label{VStarVInPiPrimeIfSigmaIsFromAut}
    Let $\Lambda$ be a cone such that $\Lambda \subseteq \Lambda_0$, and such that $\Lambda$ is disjoint from $\Lambda_3$.
    Let $\tilde{\sigma}\in \Aut(\caA)$ and assume that $\tilde{\sigma}|_{\caA_{\Lambda^c}} = \id_{\caA_{\Lambda^c}}$. Let $\sigma = \pi \circ \tilde{\sigma} \in O_{\Lambda}$. Then for all $s_1, s_2 \ge s \ge 0$ and $V_{\sigma,\Lambda_3(s_1)} \in \caV_{\sigma,\Lambda_3(s_1)}$ and $V_{\sigma,\Lambda_3(s_2)} \in \caV_{\sigma,\Lambda_3(s_2)}$, $V_{\sigma,\Lambda_3(s_2)}^*V_{\sigma,\Lambda_3(s_1)} \in \pi(\caA_{\Lambda_3(s)^c})'$.
\end{lem}
\begin{proof}
    By Lemma~\ref{VStarVInSigmaPrime}, $V_{\sigma,\Lambda_3(s_2)}^*V_{\sigma,\Lambda_3(s_1)} \in \sigma(\caA_{(\Lambda_3(s_1) \cup \Lambda_3(s_2))^c})' \subseteq \sigma(\caA_{\Lambda_3(s)^c})'$. As $\tilde{\sigma}$ is an automorphism which is the identity when restricted to $\caA_{\Lambda^c}$, it is also an automorphism when restricted to $\caA_{\Lambda}$. And, as $\Lambda_3(s)^c \supseteq \Lambda$, $\tilde{\sigma}$ is also an automorphism when restricted to $\caA_{\Lambda_3(s)^c}$. As such, $\sigma(\caA_{\Lambda_3(s)^c})=\pi(\tilde{\sigma}(\caA_{\Lambda_3(s)^c}))=\pi(\caA_{\Lambda_3(s)^c}).$ So, $V_{\sigma,\Lambda_3(s_2)}^*V_{\sigma,\Lambda_3(s_1)} \in \sigma(\caA_{\Lambda_3(s)^c})' = \pi(\caA_{\Lambda_3(s)^c})'$.   
\end{proof}

\begin{lem}\label{braidingHomSpaceForDisjointRepsSupport}
    Let $\Lambda_a, \Lambda_b \in \caC$ be disjoint subsets of $\Lambda_0$, and let $\sigma_a \in O_{\Lambda_a}$ and $\sigma_b \in O_{\Lambda_b}$.
    Let $R \in \mathrm{Hom}(\sigma_a \otimes \sigma_b, \sigma_b \otimes \sigma_a)$.
    Then $R \in \pi(\caA_{(\Lambda_a \cup \Lambda_b)^c})'$.
\end{lem}
\begin{proof}
    This proof is essentially the same as that of \cite[Lemma~4.2]{MTC}.
    For $A \in \caA_{(\Lambda_a \cup \Lambda_b)^c}$, $\sigma_a \otimes \sigma_b(A) = \pi(A) = \sigma_b \otimes \sigma_a(A)$. As for all $A \in \caA$, $R \cdot (\sigma_a \otimes \sigma_b)(A) = (\sigma_b \otimes \sigma_a)(A) \cdot R$, in particular, for all $A \in \caA_{(\Lambda_a \cup \Lambda_b)^c}$, $R \cdot \pi(A) = R \cdot (\sigma_a \otimes \sigma_b)(A) = (\sigma_b \otimes \sigma_a)(A) \cdot R = \pi(A) \cdot R$. So, $R \in \pi(\caA_{(\Lambda_a \cup \Lambda_b)^c})'$.
\end{proof}
In particular, $\epsilon(\sigma_a, \sigma_b) \in \pi(\caA_{(\Lambda_a \cup \Lambda_b)^c})'$.

\begin{lem}\label{epsilonOfRhoSigmaPrime}
    Suppose $\rho \in O_{\Lambda_1}$ and $\sigma \in O_{\Lambda_0}$, and $V \in \caU(\caH)$ such that $\sigma' = \Ad(V)\circ\sigma \in O_{\Lambda_0}$ as well. Then, $\epsilon(\rho,\sigma')=\Ad(V)(\epsilon(\rho,\sigma)) \cdot V \cdot T_\rho(V^*)$.
\end{lem}
\begin{proof}
    Picking $V_{\sigma',\Lambda_2(t)}=V_{\sigma,\Lambda_2(t)}V^*$,
    \begin{align*}
        \epsilon(\rho,\sigma')&=\lim_{t\to\infty}V_{\sigma',\Lambda_2(t)}^* T_\rho(V_{\sigma',\Lambda_2(t)})\\
        &=\lim_{t\to\infty}V V_{\sigma,\Lambda_2(t)}^* T_\rho(V_{\sigma,\Lambda_2(t)})T_\rho(V^*)\\
        &=\lim_{t\to\infty}\Ad(V)(V_{\sigma,\Lambda_2(t)}^* T_\rho(V_{\sigma,\Lambda_2(t)})) \cdot V \cdot T_\rho(V^*)\\
        &=\Ad(V)(\epsilon(\rho,\sigma))\cdot V \cdot T_\rho(V^*).
    \end{align*}
    That $V_{\sigma,\Lambda_2(t)}, V \in \caB$ follows from $V_\sigma,\Lambda_2(t) \in \pi(\caA_{(\Lambda_2(t) \cup \Lambda_0)^c})' = \pi(\caA_{\Lambda_0^c})'$ and $V \in \pi(\caA_{(\Lambda_0\cup \Lambda_0)^c})' = \pi(\caA_{\Lambda_0^c})'$ and $\pi(\caA_{\Lambda_0^c})' \subseteq \caB$, so the second equation splitting $T_\rho(V_{\sigma',\Lambda_2(t)})$ into $T_\rho(V_{\sigma,\Lambda_2(t)})T_\rho(V^*)$ is legitimate.
\end{proof}

\begin{lem}
    Let $\Lambda \subset \Lambda_0$ be disjoint from $(\Lambda_3(s_\Lambda))_{2\cdot 10^{-3}}$ for some $s_\Lambda$. Let $\rho \in O_{\Lambda_1}$ and $\sigma \in O_{\Lambda}$. Let $\sigma$ be of the form $\sigma = \pi \circ \tilde{\sigma}$ for some $\tilde{\sigma} \in \Aut(\caA)$ such that $\tilde{\sigma}|_{\caA_{\Lambda^c}} = \id_{\caA_{\Lambda^c}}$.
    For $s > 0$, let $V_{\sigma,\Lambda_3(s)} \in \caV_{\sigma,\Lambda_3(s)}$. Let $\sigma_{\Lambda_3(s)} := \Ad(V_{\sigma,\Lambda_3(s)})\circ \sigma$.
    Then, $\lim\limits_{s\to\infty} \epsilon(\rho,\sigma_{\Lambda_3(s)})$ exists, and is independent of the choice of $V_{\sigma,\Lambda_3(s)} \in \caV_{\sigma,\Lambda_3(s)}$.
\end{lem}
\begin{proof}
We start by showing that the limit exists. This will be done by showing that the sequence is Cauchy.
    By Lemma~\ref{epsilonOfRhoSigmaPrime} , $\epsilon(\rho,\sigma_{\Lambda_3(s)}) = V_{\sigma, \Lambda_3(s)} \epsilon(\rho, \sigma) T_\rho(V_{\sigma, \Lambda_3(s)}^*)$.
Using that $T_\rho(V_{\sigma, \Lambda_3(s)}^*)$ is unitary, for $s_1, s_2 > s > 0$, \begin{align*}
        \Vert\epsilon(\rho,\sigma_{\Lambda_3(s_2)}) &- \epsilon(\rho,\sigma_{\Lambda_3(s_1)})\Vert\\
        &=\norm{V_{\sigma, \Lambda_3(s_2)} \epsilon(\rho, \sigma) T_\rho(V_{\sigma, \Lambda_3(s_2)}^*) - V_{\sigma, \Lambda_3(s_1)} \epsilon(\rho, \sigma) T_\rho(V_{\sigma, \Lambda_3(s_1)}^*)}\\
        &=\norm{V_{\sigma, \Lambda_3(s_1)}^*V_{\sigma, \Lambda_3(s_2)} \epsilon(\rho, \sigma) T_\rho(V_{\sigma, \Lambda_3(s_2)}^*)T_\rho(V_{\sigma, \Lambda_3(s_1)}) -  \epsilon(\rho, \sigma)}\\
        &=\norm{V_{\sigma, \Lambda_3(s_1)}^*V_{\sigma, \Lambda_3(s_2)} \epsilon(\rho, \sigma) T_\rho(V_{\sigma, \Lambda_3(s_2)}^* V_{\sigma, \Lambda_3(s_1)}) -  \epsilon(\rho, \sigma)}.
    \end{align*}
    Let $V_{\sigma,s_2,s_1} = V_{\sigma, \Lambda_3(s_2)}^* V_{\sigma, \Lambda_3(s_1)}$ so the above becomes $$\norm{\epsilon(\rho,\sigma_{\Lambda_3(s_2)}) - \epsilon(\rho,\sigma_{\Lambda_3(s_1)})}=\norm{V_{\sigma,s_2,s_1}^* \epsilon(\rho,\sigma) T_\rho(V_{\sigma,s_2,s_1}) - \epsilon(\rho,\sigma)}.$$

    By Lemma~\ref{VStarVInPiPrimeIfSigmaIsFromAut}, $V_{\sigma,s_2,s_1} \in \pi(\caA_{\Lambda_3(s)^c})'$. 
    By Lemma~\ref{braidingHomSpaceForDisjointRepsSupport}, $\epsilon(\rho,\sigma)\in \pi(\caA_{(\Lambda_1 \cup \Lambda)^c})'$.
    By Lemma~\ref{usingApproxHaagDuality1}, for $s > 2R_{\abs{\arg \Lambda_3}, \varepsilon}$, setting $V_{\sigma,s_2,s_1,s} = \Ad(\tilde{U}_{r})(V_{\sigma,s_2,s_1})$, $V_{\sigma,s_2,s_1,s} \in \pi(\caA_{(\Lambda_3(s-\frac{s}{2}))_{\varepsilon+\delta}})''$, and $\norm{V_{\sigma,s_2,s_1}-V_{\sigma,s_2,s_1,s}}\le 2 f_{\abs{\arg \Lambda_3},\varepsilon,\delta}(\frac{s}{2})$. For concreteness, pick $\varepsilon=\delta=10^{-3}$.
    For $s > \max(2 s_\Lambda, 2 R_{\abs{\arg \Lambda_3}, 10^{-3}})$, $(\Lambda_3(\frac{s}{2}))_{2\cdot 10^{-3}}\subseteq (\Lambda_3(s_\Lambda))_{2\cdot 10^{-3}}$ and is therefore disjoint from $\Lambda$, and $(\Lambda_3(\frac{s}{2}))_{2\cdot 10^{-3}}\subseteq (\Lambda_3)_{2\cdot 10^{-3}}$ and is therefore disjoint from $\Lambda_1$, and so $(\Lambda_3(\frac{s}{2}))_{2\cdot 10^{-3}}$ is disjoint from $\Lambda_1 \cup \Lambda$.
    Therefore $[V_{\sigma,s_2,s_1,s}^*, \epsilon(\rho,\sigma)]=0$, and we decompose
    $$
        V_{\sigma,s_2,s_1}^* \epsilon(\rho,\sigma) = \epsilon(\rho,\sigma) V_{\sigma,s_2,s_1,s}^* + (V_{\sigma,s_2,s_1}-V_{\sigma,s_2,s_1,s})^*\, \epsilon(\rho,\sigma).
$$
    As $(\Lambda_3(\frac{s}{2}))_{2\cdot 10^{-3}} \in \caC$, $T_\rho$ is weak-continuous on $\pi(\caA_{(\Lambda_3(\frac{s}{2}))_{2\cdot 10^{-3}}})''$, and as $(\Lambda_3(\frac{s}{2}))_{2\cdot 10^{-3}} \subseteq \Lambda_1^c$, $T_\rho$ is the identity on $\pi(\caA_{(\Lambda_3(\frac{s}{2}))_{2\cdot 10^{-3}}})$, so together we get that it is also the identity on $\pi(\caA_{(\Lambda_3(\frac{s}{2}))_{2\cdot 10^{-3}}})''$, and so 
    $T_\rho(V_{\sigma,s_2,s_1,s})=V_{\sigma,s_2,s_1,s}$. We get,
    \begin{align*}
        T_\rho(V_{\sigma,s_2,s_1})&=
        T_\rho(V_{\sigma,s_2,s_1,s}+(V_{\sigma,s_2,s_1}-V_{\sigma,s_2,s_1,s}))\\
        &=V_{\sigma,s_2,s_1,s}+T_\rho(V_{\sigma,s_2,s_1}-V_{\sigma,s_2,s_1,s}).
    \end{align*}
    Therefore, using again that $[V_{\sigma,s_2,s_1,s}^*, \epsilon(\rho,\sigma)]=0$,
    \begin{multline*} V_{\sigma,s_2,s_1}^* \epsilon(\rho,\sigma) T_\rho(V_{\sigma,s_2,s_1}) =
    \epsilon(\rho,\sigma) 
    +(\epsilon(\rho,\sigma) V_{\sigma,s_2,s_1,s}^*)T_\rho(V_{\sigma,s_2,s_1}-V_{\sigma,s_2,s_1,s}) \\
    + (V_{\sigma,s_2,s_1}-V_{\sigma,s_2,s_1,s})^*\, \epsilon(\rho,\sigma) T_\rho(V_{\sigma,s_2,s_1}).
    \end{multline*}
    So 
    \begin{align*}
    \norm{V_{\sigma,s_2,s_1}^* \epsilon(\rho,\sigma) T_\rho(V_{\sigma,s_2,s_1}) - \epsilon(\rho,\sigma)} &\le \norm{\epsilon(\rho,\sigma) V_{\sigma,s_2,s_1,s}^*} \norm{T_\rho(V_{\sigma,s_2,s_1}-V_{\sigma,s_2,s_1,s})} \\ & \quad + \norm{(V_{\sigma,s_2,s_1}-V_{\sigma,s_2,s_1,s})^*} \norm{\epsilon(\rho,\sigma) T_\rho(V_{\sigma,s_2,s_1})} \\ 
    &\le 2 \norm{\epsilon(\rho,\sigma)} \norm{V_{\sigma,s_2,s_1}-V_{\sigma,s_2,s_1,s}}\\
    &\le 2 \cdot 2  f_{\abs{\arg \Lambda_3},10^{-3},10^{-3}}(\frac{s}{2}),
    \end{align*}
     which goes to $0$ as $s \to \infty$. Therefore, the sequence $(\epsilon(\rho,\sigma_{\Lambda_3(s)}))_{s \in \bbN}$ is Cauchy, 
and the sequence converges, i.e.,
    $$\theta(\rho,\sigma) = \lim_{s\to\infty}\epsilon(\rho,\sigma_{\Lambda_3(s)})$$
    exists.
    
Inspecting the proof, we showed that for any choice of $V_{\sigma,\Lambda_3(s)}$, we have
$$
\| \theta(\rho,\sigma)  - \epsilon(\rho,\sigma_{\Lambda_3(s)}) \| \leq 4 f_{\abs{\arg \Lambda_3},10^{-3},10^{-3}}(\frac{s}{2}).
$$
We will use this to show that the limit is independent of the choice of  $V_{\sigma,\Lambda_3(s)} \in \caV_{\sigma,\Lambda_3(s)}$.

Let $V_{\sigma,\Lambda_3(s)}, V'_{\sigma,\Lambda_3(s)}$, where for each $s$, $V_{\sigma,\Lambda_3(s)},V'_{\sigma,\Lambda_3(s)} \in \caV_{\sigma,\Lambda_3(s)}$, be two choices. Now consider a third choice, a  sequence $V''_{\sigma,\Lambda_3(s)}$ which for $s < s'$ has $V''_{\sigma,\Lambda_3(s)} = V'_{\sigma,\Lambda_3(s)}$, but for $s \ge s'$ has $V''_{\sigma,\Lambda_3(s)} = V_{\sigma,\Lambda_3(s)}$. By the above bound the limit point of the sequence, which is $\theta(\rho, \sigma)$, has a distance bounded by $4 f_{\abs{\arg \Lambda_3},10^{-3},10^{-3}}(\frac{s}{2}) $ from the limit point of the sequence corresponding to the choice $V'_{\sigma,\Lambda_3(s)}$.

So, the limit exists and is independent of the choice of $V_{\sigma,\Lambda_3(s)}$, as desired.
\end{proof}

\begin{lem}\label{thetaRhoSigmaIsHomRhoRho}
 Let $\rho \in O_{\Lambda_1}$ and $\sigma \in O_{\Lambda_0}$. Suppose that $\sigma$ be of the form $\sigma = \pi \circ \tilde{\sigma}$ for some $\tilde{\sigma} \in \Aut(\caA)$ such that $\tilde{\sigma}|_{\caA_{\Lambda^c}} = \id_{\caA_{\Lambda^c}}$. Then, $\theta(\rho,\sigma) \in \mathrm{Hom}(\rho,\rho)$.
\end{lem}
\begin{proof}
    The task is to show that for all $A \in \caA$, $\theta(\rho,\sigma) \rho(A) = \rho(A) \theta(\rho,\sigma)$. In fact, by density, it is enough to show it for $A \in \mathcal{A}_\loc$.
Let $\sigma_{\Lambda_3(t)} := \Ad(V_{\sigma,\Lambda_3(t)})\circ \sigma$. For all $t$, $\epsilon(\rho,\sigma_{\Lambda_3(t)}) \in \mathrm{Hom}(\rho \otimes \sigma_{\Lambda_3(t)}, \sigma_{\Lambda_3(t)} \otimes \rho)$. Pick $r \in \bbN$ such that $A \in \caA_{\Lambda_1(-r)}$. 
    For $t > \cot(\frac{\pi}{2}-\frac{\pi}{8}) r$, $\Lambda_1(-r) \subset \Lambda_3(t)^c$, and so $T_{\sigma_{\Lambda_3(t)}}|_{\pi(\caA_{\Lambda_1(-r)})} = \id$.
    Therefore, $\rho \otimes \sigma_{\Lambda_3(t)} (A) = T_\rho \circ T_{\sigma_{\Lambda_3(t)}}\circ \pi(A) = T_\rho\circ\pi(A)=\rho(A)$, and we get
    \begin{align*}
     \epsilon(\rho,\sigma_{\Lambda_3(t)}) \cdot \rho(A) &= \epsilon(\rho,\sigma_{\Lambda_3(t)}) \cdot (\rho \otimes \sigma_{\Lambda_3(t)}) (A)\\
     &= (\sigma_{\Lambda_3(t)} \otimes \rho)(A) \cdot \epsilon(\rho,\sigma_{\Lambda_3(t)})\\
     &= (T_{\sigma_{\Lambda_3(t)}}\circ T_\rho \circ \pi (A)) \cdot \epsilon(\rho,\sigma_{\Lambda_3(t)})\\
     &= T_{\sigma_{\Lambda_3(t)}}(\rho(A)) \cdot \epsilon(\rho,\sigma_{\Lambda_3(t)}).
    \end{align*}
    And so,
    \begin{align*}
        \theta(\rho,\sigma) \cdot \rho(A) &= \lim_{t\to\infty}\epsilon(\rho,\sigma_{\Lambda_3(t)}) \cdot \rho(A)\\
        &= \lim_{t\to\infty} T_{\sigma_{\Lambda_3(t)}}(\rho(A)) \cdot \epsilon(\rho,\sigma_{\Lambda_3(t)}).
    \end{align*}
    As for all $t$, $\epsilon(\rho,\sigma_{\Lambda_3(t)})$ is a unitary, and $\theta(\rho,\sigma) = \lim_{t\to\infty}\epsilon(\rho,\sigma_{\Lambda_3(t)})$, we conclude that $\theta(\rho,\sigma) \cdot \rho(A) \cdot \theta(\rho,\sigma)^* = \lim_{t\to\infty} T_{\sigma_{\Lambda_3(t)}}(\rho(A))$, and in particular that the limit on the right hand side exists.

    Now to conclude the proof we need to show that this limit is equal to $\rho(A)$. We will use many cones below, and we summarize their position in Figure~\ref{AppendixCones}.
    As $A \in \caA_{\Lambda_1(-r)}$, $\rho(A)=T_\rho\circ \pi(A)=\Ad( V_{\rho,K_{\Lambda_1(-r)}}^*)(\pi(A))$, where $K_{\Lambda_1(-r)}$ can be chosen to be any cone in $\caC$ which is
    \begin{enumerate}
        \item distal from $\Lambda_1(-r)$ with forbidden direction that of $\caC$ (for the definition of distal see \cite{MTC}, we will only use that such a cone exists) and
        \item clockwise between $\Lambda_1(-r)$ and the forbidden direction.
    \end{enumerate}
    For $C_r = \Lambda_1(-r) \vee \Lambda_1 \vee K_{\Lambda_1(-r)} = \Lambda_1(-r) \vee K_{\Lambda_1(-r)}$
    (the smallest cone including both $\Lambda_1(-r)$ and $K_{\Lambda_1(-r)}$ ), $\pi(\caA_{\Lambda_1(-r)})\subseteq \pi(\caA_{C_r}) \subseteq \pi(\caA_{C_r^c})'$, we have $ V_{\rho,K_{\Lambda_1(-r)}}^*  \in \pi(\caA_{(\Lambda_1 \cup K_{\Lambda_1(-r)})^c})' \subseteq \pi(\caA_{C_r^c})'$, and so $\rho(A)=T_{\rho,\Lambda_1(-r)}(\pi(A))=\Ad(V_{\rho,K_{\Lambda_1(-r)}}^*)(\pi(A)) \in \pi(\caA_{C_r^c})'$.
    Now we want to use approximate Haag duality to find elements of $\caB$ which approximate this and on which $T_{\sigma_{\Lambda_3(t)}}$ acts as the identity $T_{\sigma_{\Lambda_3(t)}}|_{\pi(\caA_{\Lambda_3(t)^c})}=\id_{\pi(\caA_{\Lambda_3(t)^c})}$. So, we want to pick $K_{\Lambda_1(-r)}$ so that we can find expanded versions of the corresponding $C_r$, to get arbitrarily good (as $t \to \infty$) approximations to $\rho(A)$ there, and where these expanded versions of $C_r$ are both elements of $\caC$ and subsets of $\Lambda_3(t)^c$. So, we want to pick an interval of directions which is a little bit clockwise of the interval of directions for $\Lambda_1(-r)$, and a basepoint, so that even after moving it back and widening it a little, it will still be disjoint from $\Lambda_1(-r) = \Lambda_{-r e_{\frac{\pi}{2}}, \frac{16\pi}{32}, \frac{4\pi}{32}}$. Choose the interval of directions for it to be $(\frac{9\pi}{32}-\frac{\pi}{32}, \frac{9\pi}{32}+\frac{\pi}{32})$. Then, for the basepoint, start with the basepoint of $\Lambda_1(-r)$ (where the only intersection would be the common basepoint), and move it forwards from there by enough to make $K_{\Lambda_1(-r)}$ distal from $\Lambda_1(-r)$.
    Specifically, let $\vec{x}_{r} = (-r) \bm{e}_{\frac{\pi}{2}} + (R_{2\frac{\pi}{32},\varepsilon}+2) \bm{e}_{\frac{9\pi}{32}}$, for $\varepsilon = \frac{\pi}{64}$, and 
    let $K_{\Lambda_1(-r)} = \Lambda_{\vec{x}_r,\frac{9\pi}{32},\frac{\pi}{32}}$.
    To check that $K_{\Lambda_1(-r)}$ is distal from $\Lambda_1(-r)$, pick $\varepsilon=\frac{\pi}{64}$ and see that as $\frac{\pi}{64} < (\frac{\pi}{2}-\frac{\pi}{8})-(\frac{9\pi}{32}+\frac{\pi}{32})$, that the range of directions for $(\Lambda_{\vec{x}_r,\frac{9\pi}{32},\frac{\pi}{32}})_{\varepsilon}$ and $\Lambda_1(-r)=\Lambda_{-r e_{\frac{\pi}{2}}, \frac{\pi}{2}, \frac{\pi}{8}}$ are disjoint, and so $(\Lambda_{\vec{x}_r,\frac{9\pi}{32},\frac{\pi}{32}} - R_{2 \\\frac{\pi}{32},\varepsilon} \vec{e}\bm{e}_{\frac{9\pi}{32}})_{\varepsilon} = \Lambda_{(-r \bm{e}_{\frac{\pi}{2}} +2 \bm{e}_{\frac{9\pi}{32}}),\frac{9\pi}{32},\frac{\pi}{32}+\varepsilon}$ is disjoint from $\Lambda_1(-r)$. From this, and that $(K_{\Lambda_1(-r)})_{\varepsilon}$ and $(\Lambda_1(-r))_\varepsilon$ are disjoint element of $\caC$, we have that $K_{\Lambda_1(-r)}$ is distal from $\Lambda_1(-r)$ with forbidden direction $(\frac{3\pi}{2}-\frac{\pi}{4},\frac{3\pi}{2}+\frac{\pi}{4})$. It is also clockwise from $\Lambda_1(-r)$ with respect to the forbidden direction. Therefore it is a valid choice for $K_{\Lambda_1(-r)}$.

\begin{center}
\begin{figure}
\includegraphics[width=0.5\textwidth]{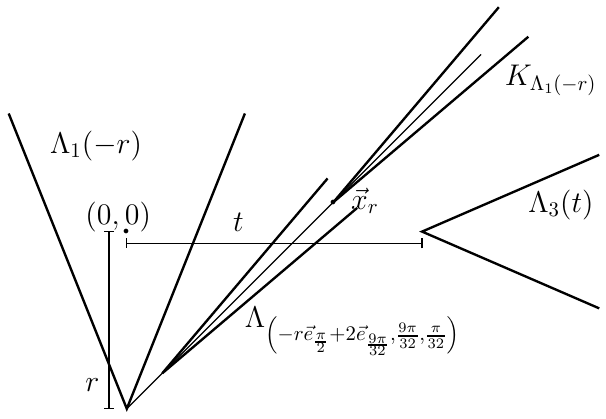}
\caption{The cones used in the proof of Lemma~\ref{thetaRhoSigmaIsHomRhoRho}}
\label{AppendixCones}
\end{figure}
\end{center}   
    With this choice of $K_{\Lambda_1(-r)}$, $C_r = \Lambda_1(-r) \vee K_{\Lambda_1(-r)} = \Lambda_{-r \bm{e}_{\frac{\pi}{2}}, \frac{7\pi}{16},\frac{3\pi}{16}}$.

    As $\rho(A) \in \pi(\caA_{C_r^c})'$, by Lemma~\ref{usingApproxHaagDuality1} , for $X_{t_2} = \Ad(\Tilde{U}_{{t_2}})(\rho(A))$ we have that, for $t_2 > R_{\abs{\arg C_r},\varepsilon}$, $X_{t_2} \in \pi(\caA_{(C_r)_{\varepsilon+\delta} - t_2 \bm{e}_{C_r}})''$ and $\norm{X_{t_2} - \rho(A)} < 2 \norm{\rho(A)} f_{\abs{\arg C_r},\varepsilon,\delta}(t_2)$.
    For $\varepsilon+\delta < \frac{\pi}{2}$, $((C_r)_{\varepsilon+\delta} - t_{2} \bm{e}_{C_r}) \in \caC$, and so $T_{\sigma_{\Lambda_3(t)}}$ is strongly continuous on $\pi(\caA_{((C_r)_{\varepsilon+\delta} - t_2 \bm{e}_{C_r})})''$. To have $((C_r)_{\varepsilon+\delta} - t_2 \bm{e}_{C_r}) \subset \Lambda_3(t)^c$, we need $\varepsilon+\delta < \frac{\pi}{8}$, and $t > \cot(\frac{\pi}{4}-(\varepsilon+\delta)) \cdot (r + t_2 \cdot (\sin(\frac{7\pi}{16})) - t_2 \cdot \cos(\frac{7\pi}{16})$ (this condition is obtained from the base point of $\Lambda_3(t)$ being to the right of the line which extends the right edge of the cone $((C_r)_{\varepsilon+\delta} - t_2 \bm{e}_{C_r})$). So, it suffices that $\varepsilon+\delta < \frac{\pi}{8}$ and $t \ge \cot(\frac{\pi}{8}) \cdot (r + t_2)$. So, we can set $t_2 = t \tan(\frac{\pi}{8}) - r$. Now having $((C_r)_{\varepsilon+\delta} - t_2 \bm{e}_{C_r}) \subset \Lambda_3(t)^c$, we have that $T_{\sigma_{\Lambda_3(t)}}$ is the identity on $\pi(\caA_{((C_r)_{\varepsilon+\delta} - t_2 \bm{e}_{C_r})})$, and so by the weak continuity is the identity on $\pi(\caA_{((C_r)_{\varepsilon+\delta} - t_2 \bm{e}_{C_r})})''$, and so $T_{\sigma_{\Lambda_3(t)}}(X_{t_2})=X_{t_2}$.
    Because both $\rho(A)$ and $X_{t_2}$ are elements of $\caB$, we have
    \begin{align*}
        T_{\sigma_{\Lambda_3(t)}}(\rho(A)) &= T_{\sigma_{\Lambda_3(t)}}(X_{t_2} + (\rho(A) - X_{t_2}))\\
        &= \rho(A) - (\rho(A) - X_{t_2}) + T_{\sigma_{\Lambda_3(t)}}(\rho(A) - X_{t_2}).
    \end{align*}
    Therefore, 
    \begin{align*}
    \norm{T_{\sigma_{\Lambda_3(t)}}(\rho(A)) - \rho(A)}  &\le \norm{(\rho(A) - X_{t_2}}+\norm{T_{\sigma_{\Lambda_3(t)}}(\rho(A) - X_{t_2})} \\ 
    &\le 4 \norm{\rho(A)} f_{\abs{\arg C_r},\varepsilon,\delta}(t_2),
    \end{align*}
     which goes to $0$ as $t$, and therefore $t_2$, goes to infinity.
\end{proof}

\begin{lem}\label{epsilonRhoPrimeSigma}
    Suppose $\rho \in O_{\Lambda_1}$, $\sigma \in O_{\Lambda_0}$, and $V \in \caU(\caH)$ is such that $\rho'=\Ad(V)\circ\rho \in O_{\Lambda_1}$.
    Then, $\epsilon(\rho',\sigma)=\lim_{t\to\infty}[[V_{\sigma,\Lambda_2(t)}^*,V]] \cdot \Ad(V)(\epsilon(\rho,\sigma))$.
\end{lem}
\begin{proof}
    \begin{align*}
        \epsilon(\rho',\sigma)&=\lim_{t\to\infty}V_{\sigma,\Lambda_2(t)}^* T_{\rho'}(V_{\sigma,\Lambda_2(t)})\\
        &=\lim_{t\to\infty}V_{\sigma,\Lambda_2(t)}^* \Ad(V)\circ T_{\rho}(V_{\sigma,\Lambda_2(t)})\\
        &=\lim_{t\to\infty}V_{\sigma,\Lambda_2(t)}^* V T_{\rho}(V_{\sigma,\Lambda_2(t)}) V^*\\
        &=\lim_{t\to\infty}(V_{\sigma,\Lambda_2(t)}^* V V_{\sigma,\Lambda_2(t)} V^*) V (V_{\sigma,\Lambda_2(t)}^* T_{\rho}(V_{\sigma,\Lambda_2(t)})) V^*\\
        &=\lim_{t\to\infty} [[V_{\sigma,\Lambda_2(t)}^*,V]] \Ad(V)(V_{\sigma,\Lambda_2(t)}^* T_{\rho}(V_{\sigma,\Lambda_2(t)}))\\
        &=\lim_{t\to\infty} [[V_{\sigma,\Lambda_2(t)}^*,V]]\cdot \Ad(V)(\epsilon(\rho,\sigma))
    \end{align*}
\end{proof}

\subsection*{Acknowledgements}
M.C. and M.F. were supported in part by the NSF under grant DMS-2407290.  S.B. was supported by NSERC of Canada.
Y.O. was supported by JSPS KAKENHI Grant Number 19K03534 and 22H01127.
She was also supported by JST CREST Grant Number JPMJCR19T2.
Part of this work was done during the visit of M.F. in Kyoto with the support of CREST. M.F. is grateful for hospitality at the Research Institute for Mathematical Sciences at Kyoto University.

\end{document}